\documentclass[%
 reprint,
 amsmath,amssymb,
 aps,pre,
floatfix,
]{revtex4-2}

\usepackage{graphicx}
\usepackage{dcolumn}
\usepackage{bm}

\usepackage{graphicx,xcolor,subfig,subfloat,amsthm,placeins,array,multirow,bbm}

\newtheorem*{proposition}{Proposition}

\renewcommand{\vec}[1]{\bm{#1}}
\newcommand{\mat}[1]{\bm{#1}}
\newcommand{\Pp}{\mathbb{P}}
\newcommand{\E}{\mathbb{E}}
\newcommand{\Ee}{\mathcal{E}}
\newcommand{\wh}[1]{\widehat{#1}}
\newcommand{\thi}{\theta_i}
\newcommand{\thj}{\theta_j}
\newcommand{\Bb}{\mathcal{B}} 
\newcommand{\NC}{\mathcal{N}}
\newcommand{\const}{\mbox{(const.)}}

\begin{document}


\title{Inference of Edge Correlations in Multilayer Networks}

\author{A. Roxana Pamfil}
\author{Sam D. Howison}%
\affiliation{Mathematical Institute, University of Oxford}

\author{Mason A. Porter}
\affiliation{Department of Mathematics, University of California Los Angeles}%
\affiliation{Mathematical Institute, University of Oxford}%

\date{\today}

\begin{abstract}

Many recent developments in network analysis have focused on multilayer networks, which one can use to encode time-dependent interactions, multiple types of interactions, and other complications that arise in complex systems. Like their monolayer counterparts, multilayer networks in applications often have mesoscale features, such as community structure. A prominent type of method for inferring such structures is the employment of multilayer stochastic block models (SBMs). A common (but {potentially} inadequate) assumption of these models is the sampling of edges in different layers independently, conditioned on the community labels of the nodes. In this paper, we relax this assumption of independence by incorporating edge correlations into an SBM-like model. We derive maximum-likelihood estimates of the key parameters of our model, and we propose a measure of layer correlation that reflects the similarity between connectivity patterns in different layers. Finally, we explain how to use correlated models for edge ``prediction'' (i.e., inference) in multilayer networks. By taking into account edge correlations, prediction accuracy improves both in synthetic networks and in a temporal network of shoppers who are connected to previously-purchased grocery products.
\end{abstract}


\maketitle


\section{Introduction}\label{sec:introduction}

A network is an abstract representation of a system in which entities called ``nodes'' interact with each other via connections called ``edges'' \cite{newman2018book}. Most traditionally, in a type of network called a ``graph'', each edge encodes an interaction between a pair of nodes. Networks arise in many domains and are useful for numerous practical problems, such as detecting bot accounts on Twitter \cite{davis2016}, finding vulnerabilities in electrical grids \cite{pagani2013}, and identifying potentially harmful interactions between drugs \cite{guimera2013}. A common feature of many networks is mesoscale (i.e., intermediate-scale) structures. Detecting such structures amounts to a type of coarse-graining, providing representations of a network that are more compact than listing all of the nodes and edges. Types of mesoscale structures include community structure \cite{fortunato2016community}, core--periphery structure \cite{csermely2013,rombach2017}, role similarity \cite{rossi2015}, and others. An increasingly popular approach for modeling and detecting such structures is by using stochastic block models (SBMs) \cite{peixoto2017bayesian}, a generative model that can produce networks with community structure or other mesoscale structures. 

For many applications of network analysis, it is important to move beyond ordinary graphs (i.e., ``monolayer networks") to examine more complicated network structures, such as collections of interrelated networks. One can study such structures through the flexible lens of multilayer networks \cite{kivela2014,bocca2014,Porter2018,nutshell2019}. Similar to monolayer networks, a multilayer network consists of a collection of ``state nodes'' that are connected pairwise by edges. A state node is a manifestation of a ``physical node'' (which we will also sometimes call simply a ``node''), which represents some entity, in a specific layer. Different layers may correspond to interactions in different time periods (yielding a temporal network), different types of relations (yielding a multiplex network), or other possibilities. As in the setting of monolayer networks, modeling and inferring mesoscale structures in multilayer networks is a prominent research area.

A key assumption of almost all existing models of multilayer networks with mesoscale structure is that edges are generated independently, conditioned on a multilayer partition \cite{bazzi2016generative,debacco2017,stanley2016,ghasemian2016,peixoto2015modeling,peixoto2015,
valles2016multilayer,taylor2016}. This independence condition applies both within each layer (which is inconsistent with the fact that real networks often include $3$-cliques and other small-scale structures) and across layers (which is inconsistent with the fact that the same nodes are often adjacent to each other in multiple layers). In this paper, we focus on relaxing the edge-independence assumption that applies to edges between the same two physical nodes in different layers. We still consider each pair of nodes independently.

In Fig.~\ref{fig:correlatedNetworkDiagram}, we show an example of a two-layer network with both strong positive and strong negative edge correlations. Incorporating such correlations into a network model is beneficial for many applications. For example, a multiplex network of air routes, where each layer corresponds to one airline, is likely to include some popular routes that appear in multiple layers (and unpopular routes may appear only in one layer)\cite{cardillo2013}. In a temporal social network, we expect people to have repeated interactions with other people \cite{aslak2018}; this is a stronger statement than just saying that they tend to interact more within the same community over time. Such edge persistence is also common in many bipartite user--item networks: shoppers tend to buy the same grocery products over time \cite{pamfil2018thesis}, customers of a music-streaming platform listen repeatedly to their favorite songs \cite{park2010}, and Wikipedia users edit specific pages several times \cite{yasseri2012}.

\begin{figure}[htb!]
	\includegraphics[width=0.35\textwidth]{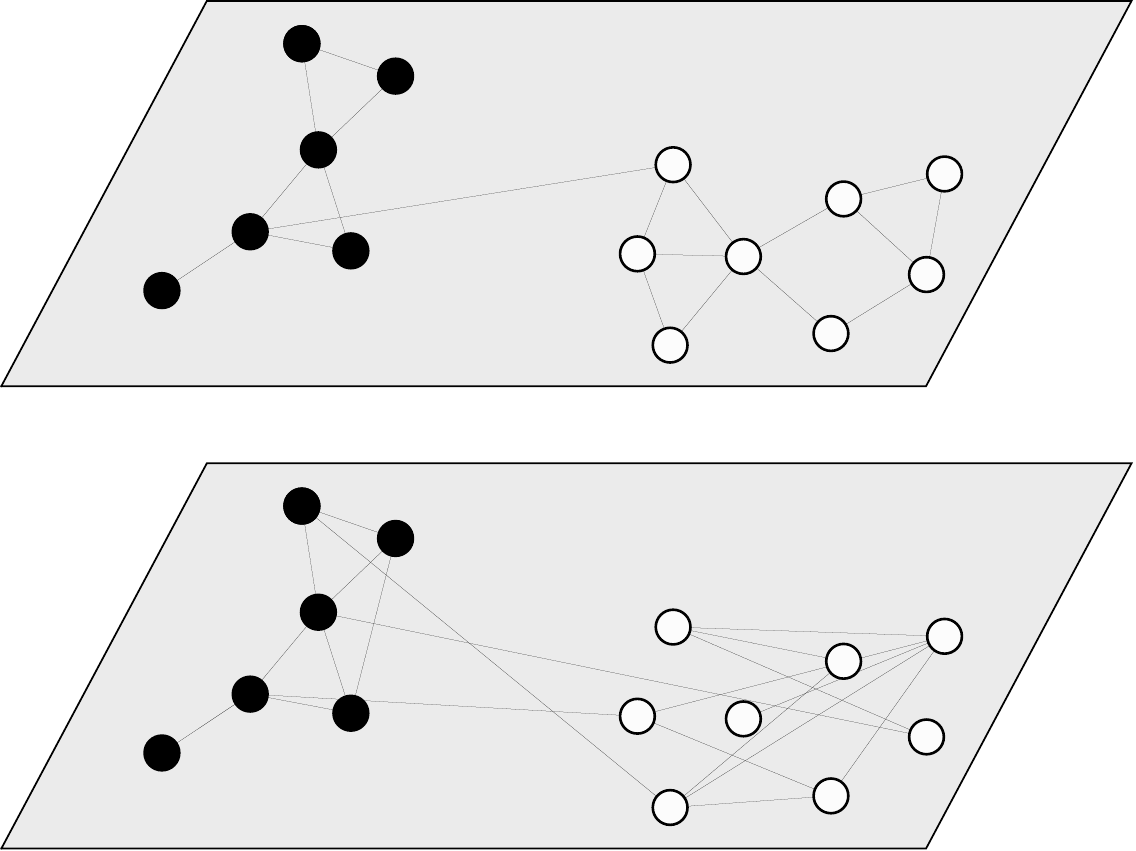}
	\caption{Example of a correlated multilayer network with two layers and with two blocks of nodes in each layer. Edges between two black nodes are positively correlated across layers, edges between two white nodes are negatively correlated across layers, and edges between a black node and a white node are uncorrelated across layers.}
	\label{fig:correlatedNetworkDiagram}
\end{figure}

Multilayer network models that incorporate edge correlations have many important applications. One is the inference task of \textit{edge prediction} (also called \textit{link prediction}), where one seeks to assign probabilities of occurrence to unobserved edges. SBMs have often been used for edge prediction for both monolayer networks \cite{aicher2014,guimera2013,peixoto2017bayesian} and multilayer networks \cite{debacco2017,valles2016multilayer,tarres2018}. 
The models that we propose should yield better performance than past efforts, as they take advantage of interlayer edge correlations in data. We discuss this in more detail in Sec.~\ref{sec:linkPrediction}. Another application is to \textit{graph matching} \cite{conte2004}, where one seeks to infer a latent correspondence between nodes in two different networks when one does not know the identities of the nodes. For example, one may wish to match common users between anonymized Twitter and Facebook networks. In a series of papers \cite{lyzinski2014,lyzinski2015,lyzinski2017}, Lyzinski and collaborators established conditions under which graph matching is successful. They tested their methods on correlated Erd\H{o}s--R\'{e}nyi (ER) networks and correlated SBMs, which we also investigate in this paper (albeit for a different purpose). We take these models further by incorporating degree correction \cite{karrer2011}, which generates networks with heterogeneous degree distributions and is important for inference using SBM. This extension may allow one to study the graph-matching problem on more realistic network models. A third application of our work is efficient computation of correlations between pairs of layers of a multilayer network. One can use such correlation estimates to quantify the similarity between different layers and potentially to compress multilayer network data by discarding (or merging) layers that are strongly positively correlated with an existing layer \cite{dedomenico2015structural}. Previous papers have focused primarily on node-centric notions of layer similarity \cite{iacovacci2015,nicosia2015,battiston2014}, whereas our correlated models yield edge-centric measures of similarity. Benefits of our approach over related studies \cite{kao2017,dedomenico2015structural,dedomenico2016spectral} include the fact that correlation values cover an intuitive range (between $-1$ and $1$) and that they work equally well for quantifying layer similarity and dissimilarity.

Our edge-correlated network models are also useful for community detection. Given a multilayer network, one can design an inference algorithm that determines both the parameters that describe the edge probabilities (and correlations) and a multilayer community structure that underlies these probabilities. Solving this inference problem enables the detection of ``correlated communities". Because of the additional complexity in the model, this is bound to be more difficult than standard multilayer community detection, so we leave this inference problem for future work. Instead, for the rest of the paper, we assume that we know the block structure of a network; we infer the remaining parameters, including the correlations that are the core element of our model \footnote{In fact, previous algorithms for SBM inference \cite{aicher2014,newman2016annotated,debacco2017,stanley2016} have suggested that one can formulate a full inference algorithm using an iterative approach that alternates between estimating a multilayer network's mesoscale structure and estimating its edge probabilities and correlations. Therefore, the latter task --- which is the focus of the present paper --- constitutes one half of an algorithm for correlated-community detection.}. We also assume that the block assignments $\vec{g}$ are the same for all layers; the case in which communities can vary arbitrarily across layers is significantly more difficult \cite{pamfil2018}, and we leave its consideration for future work. With these restrictions, one can determine $\vec{g}$ using any method of choice.

Some existing models of multilayer networks incorporate interlayer dependencies by prescribing joint degree distributions \cite{lee2012,min2014,nicosia2015}, by incorporating edge overlaps \cite{marceau2011, funk2010}, or by modeling the appearance of new edges through preferential-attachment mechanisms \cite{kim2013}. The models that we describe in this paper are similar to those that were introduced in \cite{lyzinski2014,lyzinski2015,lyzinski2017} for graph-matching purposes. Another noteworthy paper is one by Barucca et al. \cite{barucca2017} that described a generalized version of the temporal SBM of Ghasemian et al. \cite{ghasemian2016}. This generalization includes an ``edge-persistence" parameter $\xi$, which gives the probability that an edge from one layer also occurs in the next temporal layer. For several reasons, we take a different approach. First, the model of \cite{barucca2017} is specific to temporal networks, whereas we are also interested in other types of multilayer networks. Second, their model does not easily incorporate degree correction. Third, we want to include correlations explicitly in the model, rather than implicitly using the edge-persistence parameter $\xi$. 

Our paper proceeds as follows. In Sec.~\ref{sec:corrModels}, we describe our models of multilayer networks with edge correlations. We start with a simple example of correlated Erd\H{o}s--R\'{e}nyi (ER) graphs in Sec.~\ref{subsec:ER} to make our exposition for more complicated models easier to follow. In Sec.~\ref{subsec:SBM}, we integrate mesoscale structures by incorporating correlations in an SBM-like model. We then introduce degree correction in Sec.~\ref{subsec:DCSBM}. For all of these models, we derive maximum-likelihood (ML) estimates both of the marginal edge-existence probabilities in each layer and of the interlayer correlations. {ML estimation is common for SBMs and DCSBMs in both monolayer networks \cite{karrer2011} and multilayer networks \cite{stanley2016}. Although ML estimation is less powerful than performing Bayesian inference \cite{peixoto2017bayesian}, the former is consistent for both SBMs and DCSBMs \cite{zhao2012} and it recovers many common techniques for detecting mesoscale structures in networks \cite{young2018}.} In Sec.~\ref{sec:linkPrediction}, we describe how to use our models for edge prediction, and we provide some results for synthetic networks. We then proceed with applications in Sec.~\ref{sec:applications}. In Sec.~\ref{subsec:layerCorr}, we use our models to estimate pairwise layer correlations in several empirical networks. In Sec.~\ref{subsec:shopping}, we use our correlated models for edge prediction in a temporal network of grocery purchases. We summarize our results in Sec.~\ref{sec:conclusions} and discuss a few ideas for future work.


\section{Correlated Models}\label{sec:corrModels}

In our derivations, we consider just two network layers at a time. Although this may seem limiting, we can apply our framework to generate correlated networks with more than two layers in a sequential manner (see the discussion in Sec. \ref{subsec:ER}), and we can determine pairwise layer correlations for a network with arbitrarily many layers (see the applications in Sec.~\ref{sec:applications}). In Sec.~\ref{sec:conclusions}, we briefly discuss the challenges that arise when modeling three or more layers simultaneously, rather than in a sequential pairwise fashion. 

{Consider a network with two layers and identical sets of nodes in each layer; this is known as a \textit{node-aligned} multilayer network \cite{kivela2014}.}
Let $\mat{A}^1$ and $\mat{A}^2$ denote the adjacency matrices of our two network layers. As in many generative models of networks, we assume that edges in these two layers are generated by some random process, so the entries $A_{ij}^1$ and $A_{ij}^2$ are random variables. Imposing some statistical correlation between these two sets of random variables introduces interlayer correlations in the resulting multilayer network structure. 

Our goal is to propose a model of correlated networks in which each layer is, marginally, a degree-corrected stochastic block model (DCSBM) \cite{karrer2011}. However, it is instructive to first consider the simpler cases in which each layer is marginally an Erd\H{o}s--R\'{e}nyi random graph (see Sec.~\ref{subsec:ER}) or an SBM without degree correction (see Sec.~\ref{subsec:SBM}). Correlated ER models and correlated SBMs have been studied previously, most notably in work by Lyzinski et al. \cite{lyzinski2014,lyzinski2015,lyzinski2017} on the graph-matching problem. However, our use of these models for estimating layer correlations is novel, as are the correlated DCSBMs that we propose in Sec.~\ref{subsec:DCSBM}.

In monolayer SBMs, it is common to use either Bernoulli or Poisson random variables to generate edges between nodes. The former is generally more accurate, because it does not yield multiedges; however, the latter is more common, as it often simplifies calculations considerably \cite{zhao2012,perry2012}. Nevertheless, we have found that Bernoulli models are simpler when incorporating correlations \cite{pamfil2018thesis}. They also have several other advantages, including the fact that they work for both sparse and dense networks and that they can handle the entire correlation range between $-1$ and $1$. Therefore, we consider only Bernoulli models in this paper.


\subsection{Correlated Erd\H{o}s--R\'{e}nyi Layers}\label{subsec:ER}


\subsubsection{Forward model}

{Assume that the intralayer networks that correspond to $\mat{A}^1$ and $\mat{A}^2$} are ER graphs from the $\mathcal{G}(n,p)$ ensemble \cite{newman2018book} with edge probabilities $p_1$ and $p_2$. 
For each pair of nodes $(i,j)$, we therefore have
\begin{align}
	\Pp(A_{ij}^1=1) &= p_1\,, \\
	\Pp(A_{ij}^2=1) &= p_2\,.
\end{align}
To couple edges that connect the same pair of nodes in different layers, let 
\begin{equation}
	q:= \Pp(A_{ij}^1=1,A_{ij}^2=1)
\end{equation}
denote the joint probability for an edge to occur in both layers. Unless $q=p_1p_2$, this construction implies that the random variables $A_{ij}^1$ and $A_{ij}^2$ are not independent. 

The parameters $p_1$, $p_2$, and $q$ (which lie in the interval $[0,1]$) fully specify a forward model of networks with correlated ER layers. To generate a network from this model, one considers each node pair $(i,j)$ and, independently of all other node pairs, assigns values to $A_{ij}^1$ and $A_{ij}^2$ according to the following probabilities:
\begin{align} 
	\Pp(A_{ij}^1=1,A_{ij}^2=1) &= q \,, \notag \\
	\Pp(A_{ij}^1=1,A_{ij}^2=0) &= p_1-q \,, \notag \\
	\Pp(A_{ij}^1=0,A_{ij}^2=1) &= p_2-q \,, \notag \\
	\Pp(A_{ij}^1=0,A_{ij}^2=0) &= 1-p_1-p_2+q  \,.   \label{eqn:corrER1-4}
\end{align}
These expressions follow from the laws of probability and from the definitions of $p_1$, $p_2$, and $q$.
For these probabilities to be well-defined, it is both necessary and sufficient that $0 \leq q \leq \min(p_1,p_2)$ and $p_1+p_2 \leq 1 + q$. In Fig.~\ref{fig:correlatedDiagram}, we illustrate the feasible region for $p_1$ and $p_2$, given a value of $q$.

\begin{figure}[ht!]
\centering
\includegraphics[width=0.35\textwidth]{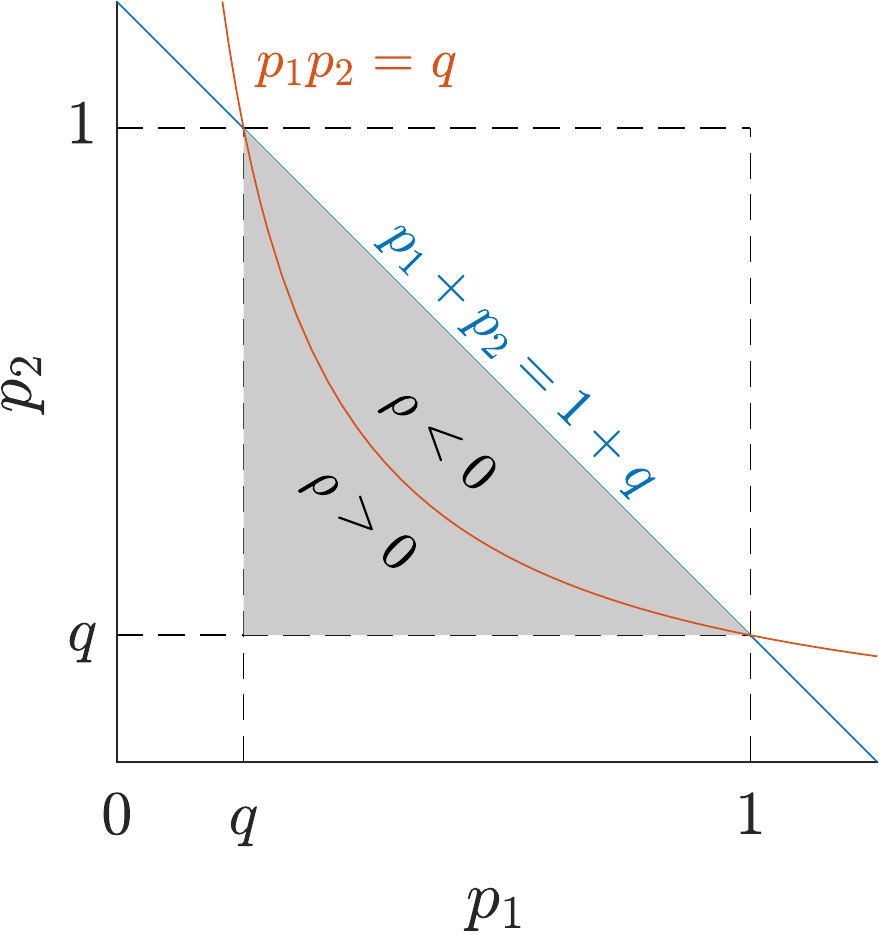}
\caption{Visualization of the feasible region (gray area) for $p_1$ and $p_2$, given a value of $q$. The boundaries of this region are defined by the inequalities $q \leq p_1 \leq 1$, $q \leq p_2 \leq 1$, and $p_1+p_2 \leq 1 + q$. The hyperbola $p_1p_2=q$ specifies the boundary between regimes with a positive layer correlation and regimes with a negative layer correlation.}
\label{fig:correlatedDiagram}
\end{figure}

It is also possible to generate a correlated ER network in a sequential manner. First, one generates the adjacency matrix $\mat{A}^1$ by placing edges with probability $p_1$. One then determines the probabilities of edges in the second layer by conditioning on the first layer:
\begin{align}
	\Pp(A_{ij}^2=1|A_{ij}^1=1) &= \frac{\Pp(A_{ij}^1=1,A_{ij}^2=1)}{\Pp(A_{ij}^1=1)} = \frac{q}{p_1} \,, \notag \\
	\Pp(A_{ij}^2=1|A_{ij}^1=0) &= \frac{\Pp(A_{ij}^1=0,A_{ij}^2=1)}{\Pp(A_{ij}^1=0)} = \frac{p_2-q}{1-p_1} \,, \notag \\
	\Pp(A_{ij}^2=0|A_{ij}^1=1) &= \frac{\Pp(A_{ij}^1=1,A_{ij}^2=0)}{\Pp(A_{ij}^1=1)} = \frac{p_1-q}{p_1} \,, \notag \\
	\Pp(A_{ij}^2=0|A_{ij}^1=0) &= \frac{\Pp(A_{ij}^1=0,A_{ij}^2=0)}{\Pp(A_{ij}^1=0)} \notag \\
		&= \frac{1-p_1-p_2+q}{1-p_1}  \,.   \label{eqn:corrERseq1-4}
\end{align}
With this approach, it is possible to generate networks with arbitrarily many layers by first sampling edges in the first layer, and then sampling edges in each subsequent layer by conditioning on the previous one. 
This kind of process is especially well-suited to temporal networks, in which layers have a natural ordering. For multiplex networks, it is more appropriate to extend Eqns. \eqref{eqn:corrER1-4} to handle more than two layers.

It is also possible to parametrize correlated ER graphs in terms of the marginal Bernoulli probabilities, $p_1$ and $p_2$, and the Pearson correlation
\begin{align}\label{eqn:rhoERBernoulli}
	\rho&=\frac{\E\left[A_{ij}^1A_{ij}^2\right]-\E\left[A_{ij}^1\right]\E\left[A_{ij}^2\right]}{\sigma[A_{ij}^1]\sigma[A_{ij}^2]} \notag \\
		&=\frac{q-p_1p_2}{\sqrt{p_1(1-p_1)p_2(1-p_2)}}\,,
\end{align} 
where $\E[\cdot]$ and $\sigma[\cdot]$, respectively, denote the mean and standard deviation of a random variable. One benefit of using $\rho$, rather than $q$, as the third model parameter is that its value is easier to interpret. A value of $\rho$ that is close to $0$ indicates a weak correlation between layers, whereas values that are close to the extremes of $+1$ and $-1$ indicate a strong positive correlation and a strong negative correlation, respectively.

We can gain further intuition by considering the cases $\rho=0$, $\rho=1$, and $\rho=-1$. First, $\rho=0$ if and only if
$\Pp(A_{ij}^1=1,A_{ij}^2=1)=\Pp(A_{ij}^1=1)\Pp(A_{ij}^2=1)$. 
That is, the correlation is $0$ if and only if edges are generated independently in the two layers with marginal probabilities of $p_1$ and $p_2$.
For $\rho=1$, one can show (see \cite{pamfil2018thesis}) that $p_1=p_2=q$, which corresponds to the two layers having identical network structure. Lastly, for $\rho=-1$, we have $q=0$ and $p_1=1-p_2$ (see \cite{pamfil2018thesis}), and two nodes are adjacent in one layer if and only if they are not adjacent in the other layer.


\subsubsection{Maximum-likelihood parameter estimates}

We now derive ML estimates of the parameters $p_1$, $p_2$, and $q$. Let $\mathcal{E}$ denote the set of node pairs that can form edges. For undirected networks without self-edges, there are $|\mathcal{E}|=N(N-1)/2$ such node pairs to consider, where $N$ is the number of physical nodes. By contrast, $|\mathcal{E}|=N(N-1)$ when generating directed networks without self-edges. With this general notation, all our derivations in Sec.\,\ref{sec:corrModels} are valid for both directed and undirected networks, with or without self-edges. (They are also valid for bipartite networks \cite{pamfil2018thesis}.) We consider each pair of nodes $(i,j) \in \mathcal{E}$ independently when generating edges, so the likelihood of observing adjacency matrices $\mat{A}^1$ and $\mat{A}^2$ is
\begin{widetext}
\begin{align}\label{eqn:corrERlikelihood}
	\Pp(\mat{A}^1,\mat{A}^2|p_1,p_2,q) =\prod_{(i,j) \in \Ee}{q^{A_{ij}^1A_{ij}^2}(p_1-q)^{A_{ij}^1(1-A_{ij}^2)}(p_2-q)^{(1-A_{ij}^1)A_{ij}^2}(1-p_1-p_2+q)^{(1-A_{ij}^1)(1-A_{ij}^2)}}  \,.
\end{align}
\end{widetext}

It is helpful to introduce the following notation:
\begin{align*}
	e_{11} ~&:=~ \lvert\{(i,j) \in \Ee: A_{ij}^1=1, A_{ij}^2=1\}\rvert \,, \\
	e_{10} ~&:=~ \lvert\{(i,j) \in \Ee: A_{ij}^1=1, A_{ij}^2=0\}\rvert \,, \\
	e_{01} ~&:=~ \lvert\{(i,j) \in \Ee: A_{ij}^1=0, A_{ij}^2=1\}\rvert \,, \\
	e_{00} ~&:=~ \lvert\{(i,j) \in \Ee: A_{ij}^1=0, A_{ij}^2=0\}\rvert \,.
\end{align*}
These quantities correspond, respectively, to the number of node pairs that are adjacent in both layers, are adjacent in the first layer but not in the second, are adjacent in the second layer but not in the first, and are not adjacent in either layer.
Using this notation and taking the logarithm of \eqref{eqn:corrERlikelihood}, we arrive at the following expression for the log-likelihood:
\begin{align}\label{eqn:corrERloglikelihood}
	\mathcal{L}&=e_{11}\log q+e_{10}\log(p_1-q)+e_{01}\log(p_2-q) \notag \\
&\qquad+e_{00}\log(1-p_1-p_2+q)\,.
\end{align}

When fitting our model to network data, the quantities $e_{11}$, $e_{10}$, $e_{01}$, $e_{00}$ are all known; and we seek to determine the values of $p_1$, $p_2$, and $q$ that are best explained by the data. To do so, we maximize the log-likelihood \eqref{eqn:corrERloglikelihood} by setting its partial derivatives to $0$ \footnote{The Hessian is negative definite at the critical point that one obtains by setting these derivatives to $0$, so this point is a local maximum (as opposed to a local minimum or saddle point) of the log-likelihood function.}. 
We obtain
\begin{align}
	\widehat{p}_1 &= \frac{e_{11}+e_{10}}{e_{11}+e_{10}+e_{01}+e_{00}}\,, \label{eqn:corrERp1} \\ 
	\widehat{p}_2 &= \frac{e_{11}+e_{01}}{e_{11}+e_{10}+e_{01}+e_{00}}\,, \label{eqn:corrERp2} \\
	\widehat{q}   &= \frac{e_{11}}{e_{11}+e_{10}+e_{01}+e_{00}} \label{eqn:corrERq} \,.
\end{align}
In all three expressions, the denominator is equal to the number of potential edges (i.e., the cardinality of $\Ee$).  Additionally, let $m_1=e_{11}+e_{10}$ and $m_2=e_{11}+e_{01}$ denote the number of observed edges in the first and the second layers, respectively. It follows that the ML estimate $\widehat{p}_1$ is equal to the number of observed edges in layer 1 divided by the number of potential edges, and an analogous relation holds for $\widehat{p}_2$. The estimate $\widehat{q}$ is equal to the number of node pairs that are adjacent in both layers divided by the total number of node pairs. These results match our intuition.

We obtain an estimate of the Pearson correlation $\rho$ between the two layers by substituting the ML estimates $\widehat{p}_1$, $\widehat{p}_2$, and $\widehat{q}$ into Eqn.~\eqref{eqn:rhoERBernoulli} to obtain 
\begin{equation}\label{eqn:rhoERBernoulliMLE}
	\widehat{\rho}=\frac{e_{00}e_{11}-e_{10}e_{01}}{\sqrt{(e_{11}+e_{10})(e_{11}+e_{01})(e_{10}+e_{00})(e_{01}+e_{00})}}\,.
\end{equation}
One can show that maximizing the log-likelihood \eqref{eqn:corrERloglikelihood} with respect to $p_1$, $p_2$, and $\rho$ (rather than with respect to $p_1$, $p_2$, and $q$) gives the same expression for $\wh{\rho}$, confirming that this is indeed an ML estimate of the correlation. Note that $\wh{\rho}$ is not defined when either layer is an empty or a complete graph, as the corresponding Bernoulli random variable has a standard deviation of $0$. 

{In App.~\ref{app:variance}, we calculate the variances that are associated with the ML estimates $\widehat{p}_1$, $\widehat{p}_2$, $\widehat{q}$. We then show using a synthetic example that these scale as $1/N^2$, where we recall that $N$ is the number of physical nodes. These results quantify the uncertainty around the ML estimates for correlated ER models as a function of network size.}


\subsection{Correlated SBMs}\label{subsec:SBM}

One of the ways in which real-world networks differ from ER random graphs is that the former have mesoscale structures, such as communities \cite{fortunato2016community}. We use SBMs to incorporate such structures into our correlated models. 

Let $\vec{g}$ be a vector of block assignments, which we take to be identical for both network layers, and let $K$ denote the number of blocks. As we explained in Sec.~\ref{sec:introduction}, we assume throughout the present paper that we are given $\vec{g}$, and we aim to estimate the remaining model parameters. Following terminology from \cite{aicher2014}, let $\Bb=\{1,\ldots,K\} \times \{1,\ldots,K\}$ be the set of ``edge bundles'' $(r,s)$, each of which is described by its own set of parameters $p_{rs}^1$, $p_{rs}^2$, and $q_{rs}$. The $K \times K$ matrices $\mat{p}^1$, $\mat{p}^2$, and $\mat{q}$ play an analogous role to $p_1$, $p_2$, and $q$ in the ER layers. 

Let $g_i$ denote the block assignment of node $i$. A correlated two-layer SBM is described by the following set of equalities: 
\begin{align*}
	\Pp(A_{ij}^1=1) = p^1_{g_ig_j} \,, \\
	\Pp(A_{ij}^2=1) = p^2_{g_ig_j} \,, \\
	\Pp(A_{ij}^1=1,A_{ij}^2=1) = q_{g_ig_j} \,.
\end{align*} 
Lyzinski et al. proposed this forward model in \cite{lyzinski2015} to study the graph-matching problem. By contrast, we focus on the inverse problem of estimating the parameters $\mat{p}^1$, $\mat{p}^2$, and $\mat{q}$, given some network data.


\subsubsection{Maximum-likelihood parameter estimates}

As in Sec.~\ref{subsec:ER}, suppose that we consider each node pair $(i,j)$ independently. The likelihood of observing adjacency matrices $\mat{A}^1$ and $\mat{A}^2$ is then
\begin{widetext}
\begin{align}
	\Pp(\mat{A}^1,\mat{A}^2|\vec{g},\vec{p}^1,\vec{p}^2,\vec{q})= \prod_{(i,j) \in \Ee}&\Bigl[q_{g_ig_j}^{A_{ij}^1A_{ij}^2}(p_{g_ig_j}^1-q_{g_ig_j})^{A_{ij}^1(1-A_{ij}^2)}(p_{g_ig_j}^2-q_{g_ig_j})^{(1-A_{ij}^1)A_{ij}^2} \notag \\
	&\times (1-p_{g_ig_j}^1-p_{g_ig_j}^2+q_{g_ig_j})^{(1-A_{ij}^2)(1-A_{ij}^2)}\Bigr]\,.
\end{align}
\end{widetext}
In this product, each factor depends on $i$ and $j$ only via their block memberships $g_i$ and $g_j$, so we can combine several terms.
First, define
\begin{equation*}
	e_{rs}^{ab}:=\left\lvert\{(i,j) \in \Ee: A_{ij}^1=a,A_{ij}^2=b,g_i=r,g_j=s\}\right\rvert 
\end{equation*}
for $(a,b) \in \{(1,1),(1,0),(0,1),(0,0)\}$, in analogy with $e_{11}$, $e_{10}$, $e_{01}$, and $e_{00}$ from Sec. \ref{subsec:ER}. We can then write the log-likelihood as
\begin{align}\label{eqn:SBMBernoulliLogL}
	\mathcal{L} &= \sum_{(r,s) \in \Bb}\Bigl[e_{rs}^{11}\log q_{rs} \notag \\
		&\qquad\qquad\quad+e_{rs}^{10}\log(p^1_{rs}-q_{rs})+e_{rs}^{01}\log(p^2_{rs}-q_{rs}) \notag \\
		&\qquad\qquad\quad+ e_{rs}^{00}\log(1-p^1_{rs}-p^2_{rs}+q_{rs}) \Bigr]\,.
\end{align}

The advantage of writing the log-likelihood as in \eqref{eqn:SBMBernoulliLogL} is that it clearly separates the contribution from different edge bundles. Using the results for ER layers from Sec.~\ref{subsec:ER}, we immediately obtain (without further calculations) the following ML parameter estimates:
\begin{align}
	\widehat{p}^1_{rs}&=\frac{e_{rs}^{11}+e_{rs}^{10}}{e_{rs}^{11}+e_{rs}^{10}+e_{rs}^{01}+e_{rs}^{00}} =\frac{m_{rs}^1}{e_{rs}}\,, \label{eqn:MLESBMBernoullip1} \\
	\widehat{p}_{rs}^2&=\frac{e_{rs}^{11}+e_{rs}^{01}}{e_{rs}^{11}+e_{rs}^{10}+e_{rs}^{01}+e_{rs}^{00}} =\frac{m_{rs}^2}{e_{rs}}\,, \label{eqn:MLESBMBernoullip2} \\ 
	\widehat{q}_{rs}&=\frac{e_{rs}^{11}}{e_{rs}^{11}+e_{rs}^{10}+e_{rs}^{01}+e_{rs}^{00}} = \frac{e_{rs}^{11}}{e_{rs}}\,, \label{eqn:MLESBMBernoulliq}
\end{align}
where $m_{rs}^1$ and $m_{rs}^2$ denote the number of edges between blocks $r$ and $s$ in layers $1$ and $2$, respectively, and $e_{rs}$ is the number of possible edges between nodes in block $r$ and block $s$. When there is a single edge bundle (i.e., when we do not assume any block structure in a network), the ML estimates \eqref{eqn:MLESBMBernoullip1}--\eqref{eqn:MLESBMBernoulliq} recover those that we obtained for correlated ER networks in Sec.~\ref{subsec:ER}. Each edge bundle also has a corresponding Pearson correlation, whose ML estimate is
\begin{equation}\label{eqn:rhoSBMBernoulli}
	\widehat{\rho}_{rs}=\frac{e_{rs}^{00}e_{rs}^{11}-e_{rs}^{10}e_{rs}^{01}}{\sqrt{(e_{rs}^{11}+e_{rs}^{10})(e_{rs}^{11}+e_{rs}^{01})(e_{rs}^{10}+e_{rs}^{00})(e_{rs}^{01}+e_{rs}^{00})}}\,.
\end{equation}
In applications to temporal consumer--product networks, we find that different edge bundles have vastly different correlation values \cite{pamfil2018thesis}. We anticipate that other empirical multilayer networks have similar properties.


\subsubsection{Effective correlation}

Although having different correlation values for different edge bundles can be useful, it is also helpful to have a single correlation measurement for a given multilayer network. 
For example, one may wish to use such a network diagnostic for one of the purposes that we outlined in Sec.~\ref{sec:introduction}. One way to define an ``effective correlation" is to first sample two node indices, $I$ and $J$, uniformly at random and then compute the Pearson correlation of the random variables $A_{IJ}^1$ and $A_{IJ}^2$. That is,
\begin{equation}\label{eqn:rhoSBMBernoulliEffFormula}
	\mathrm{corr}(A_{IJ}^1,A_{IJ}^2)=\frac{\E[A_{IJ}^1A_{IJ}^2]-E[A_{IJ}^1]\E[A_{IJ}^2]}{\sigma[A_{IJ}^1]\sigma[A_{IJ}^2]}\,,
\end{equation} 
where we use capital letters for the node indices $I$ and $J$ to emphasize that they are random variables. 

We can calculate each term on the right-hand side of \eqref{eqn:rhoSBMBernoulliEffFormula} by conditioning on the block assignments of the randomly chosen nodes $I$ and $J$. First, for $l \in \{1,2\}$, we have
\begin{align}\label{this}
	&\E[A_{IJ}^l]=\Pp(A_{IJ}^l=1) \notag \\
	&\quad=\sum_{(r,s) \in \Bb}{\Pp(A_{IJ}^l=1|g_I=r,g_J=s)\Pp(g_I=r,g_J=s)} \notag \\
	&\quad=\sum_{(r,s) \in \Bb}{\wh{p}_{rs}^1\frac{e_{rs}}{|\Ee|}}
=\sum_{(r,s) \in \Bb}{\frac{m_{rs}^l}{e_{rs}}\frac{e_{rs}}{|\Ee|}}=\frac{m_l}{|\Ee|}\,,
\end{align}
where $m_l$ denotes the number of edges in layer $l$. The expression \eqref{this} is the same as the probability $p_l$ of generating an edge in layer $l$ for the ER case. Because $A_{IJ}^l$ is a Bernoulli random variable (in other words, it can only take the values $1$ or $0$), its standard deviation is
\begin{equation*}
	\sigma[A_{IJ}^l]=\sqrt{\frac{m_l}{|\Ee|}\left(1-\frac{m_l}{|\Ee|}\right)}\,.
\end{equation*}
Lastly, 
\begin{align*}
	\E[A_{IJ}^1A_{IJ}^2]&=\Pp(A_{IJ}^1=1,A_{IJ}^2=1)\\
		&=\sum_{(r,s) \in \Bb}{\wh{q}_{rs}\frac{e_{rs}}{|\Ee|}}
	=\sum_{(r,s) \in \Bb}{\frac{e_{rs}^{11}}{e_{rs}}\frac{e_{rs}}{|\Ee|}}
	=\frac{e_{11}}{|\Ee|}\,.
\end{align*}
The estimated value of the effective correlation is thus
\begin{align}\label{eqn:rhoSBMBernoulliEff}
	\widehat{\rho}&=\mathrm{corr}(A_{IJ}^1,A_{IJ}^2)  \\
		&=\frac{e_{00}e_{11}-e_{10}e_{01}}{\sqrt{(e_{11}+e_{10})(e_{11}+e_{01})(e_{10}+e_{00})(e_{01}+e_{00})}}\,, \notag
\end{align}
which recovers the value in \eqref{eqn:rhoERBernoulliMLE} for ER layers (i.e., without any block structure in the model). 
We stress that there is no reason a priori to expect this outcome. In fact, the analogous result does \emph{not} hold for Poisson models \cite{pamfil2018thesis}. In the present case, the fact that there is such a correspondence between models is convenient for practical reasons, as it implies that one can perform the simpler calculations from Sec.~\ref{subsec:ER} to obtain correlation estimates between network layers, even for networks with nontrivial mesoscale structure.


\subsection{Correlated Degree-Corrected SBMs}\label{subsec:DCSBM}

The models that we have discussed thus far generate networks in which nodes in the same block have the same expected degree. SBMs that make this kind of assumption tend to perform poorly when they are used to infer mesoscale structure in real networks, many of which have highly heterogeneous degree distributions. This observation led to the development of degree-corrected SBMs (DCSBMs) \cite{karrer2011}. We expect that such adjustments can also make a difference when modeling edge correlations, so we now extend the model from Sec.~\ref{subsec:SBM} to incorporate degree correction. 

We continue to work with two-layer networks, which we again specify in terms of two intralayer adjacency matrices, $\mat{A}^1$ and $\mat{A}^2$, with a common block structure that we specify with a vector $\vec{g}$. For each node pair $(i,j) \in \Ee$, we place edges in the two layers according to the probabilities
\begin{align}
	\Pp(A_{ij}^1=1) &= \theta_i^1\theta_j^1p^1_{g_ig_j}\,, \label{eqn:corrDCSBMp1} \\
	\Pp(A_{ij}^2=1) &= \theta_i^2\theta_j^2p^2_{g_ig_j}\,, \label{eqn:corrDCSBMp2} \\
	\Pp(A_{ij}^1=1,A_{ij}^2=1) &= \sqrt{\theta_i^1\theta_j^1\theta_i^2\theta_j^2}q_{g_ig_j}\,. \label{eqn:corrDCSBMq}
\end{align}
We will soon justify the expression in \eqref{eqn:corrDCSBMq}. The quantities $\theta_i^l$ and $\theta_j^l$, with $l \in \{1,2\}$, are the degrees of nodes $i$ and $j$, normalized by the mean degrees. We calculate these quantities directly from an input degree sequence, so they are not model parameters. For undirected and unipartite networks, $\theta_i^l=d_i^l \,/\, \langle d^{\,l}\rangle$, where $i \in \NC$ and $\langle d^{\,l}\rangle$ is the mean degree in layer $l$.
This normalization recovers the model in Sec.~\ref{subsec:SBM} when $\theta_i^l=1$ (i.e., when all nodes have the same degree). {The model parameters $p_{rs}^1$, $p_{rs}^2$, and $q_{rs}$ are now edge ``propensities" that, together with the degrees, control the probabilities of edges in the layers.}

The probabilities in Eqns. \eqref{eqn:corrDCSBMp1}--\eqref{eqn:corrDCSBMp2} ensure that, marginally, $\mat{A}^1$ and $\mat{A}^2$ are generated according to monolayer DCSBMs \cite{karrer2011}. It is not obvious how to model the joint probability $\Pp(A_{ij}^1=1,A_{ij}^2=1)$. In particular, it is not clear how it should depend on the observed degrees of nodes  $i$ and $j$ in layers $1$ and $2$ \footnote{One desiderata of a correlated degree-corrected model is that it satisfies the inequality $\Pp(A^1_{ij} = 1, A^2_{ij} = 1) \leq \Pp(A^1_{ij} = 1)$ for all node pairs $(i,j)$. This requirement rules out some possibilities for $\Pp(A^1_{ij} = 1, A^2_{ij} = 1)$, for instance one where this quantity has no dependence on node degrees.}. 
Part of the complication is that there are four such quantities for each node pair $(i,j)$. The choice from \eqref{eqn:corrDCSBMq} works particularly well when $\rho=1$ and the normalized degree sequences $\vec{\theta}^1$ and $\vec{\theta}^2$ are the same, as it reduces to a single degree-corrected SBM that generates two identical network layers. Another sensible option is to set $\Pp(A_{ij}^1=1,A_{ij}^2=1)=\theta_i^1\theta_j^1\theta_i^2\theta_j^2q_{g_ig_j}$. 
This choice has the nice property that edges in a particular edge bundle $(r,s) \in \Bb$ are independent if and only if $q_{rs}=p_{rs}^1p_{rs}^2$, which matches the independence condition from Sec.~\ref{subsec:SBM} for the setting without degree correction. However, this second model underperforms the one from \eqref{eqn:corrDCSBMp1}--\eqref{eqn:corrDCSBMq} for edge prediction (see Sec.~\ref{sec:linkPrediction} and \cite{pamfil2018thesis}). 
Consequently, for the rest of this paper, we use the model from Eqns. \eqref{eqn:corrDCSBMp1}--\eqref{eqn:corrDCSBMq} as our correlated DCSBM. 


\subsubsection{Maximum-likelihood parameter estimates}

When writing the log-likelihood for correlated DCSBMs, we can ignore any additive terms that only involve known quantities, such as the normalized degrees $\theta_i^l$. We can thus write
\begin{widetext}
\begin{align}\label{eqn:DCSBMlogL} 
	\mathcal{L}=\sum_{(r,s) \in \Bb}\sum_{(i,j) \in \Ee}&\Biggl[A_{ij}^1A_{ij}^2\log q_{rs} + A_{ij}^1(1-A_{ij}^2)\log\left(p_{rs}^1-\sqrt{\frac{\thi^2\thj^2}{\thi^1\thj^1}}q_{rs}\right) +(1-A_{ij}^1)A_{ij}^2\log\left(p_{rs}^2-\sqrt{\frac{\thi^1\thj^1}{\thi^2\thj^2}}q_{rs}\right) \notag \\
	& +(1-A_{ij}^1)(1-A_{ij}^2)\log\left(1-\theta_i^1\theta_j^1p^1_{rs}-\theta_i^2\theta_j^2p^2_{rs}+\sqrt{\theta_i^1\theta_j^1\theta_i^2\theta_j^2}q_{rs}\right)  \Biggr]\delta(g_i,r)\delta(g_j,s) +\const\,.
\end{align}
\end{widetext}
As in Sec.~\ref{subsec:SBM}, we seek to maximize $\mathcal{L}$ with respect to the parameters $p_{rs}^1$, $p_{rs}^2$, and $q_{rs}$ by setting the corresponding derivatives to $0$. However, degree-corrected models have the crucial complication that node pairs $(i,j)$ in the same edge bundle $(r,s)$ are no longer stochastically equivalent (i.e., the corresponding entries of the adjacency matrix are no longer sampled from independent, identically distributed random variables), so their contributions to the log-likelihood are no longer the same in general.
Consequently, the ML equations for correlated DCSBMs involve $\mathcal{O}(N^2/K^2)$ terms, making them more difficult to solve efficiently. 

{In certain cases,} we are able to make some approximations that make these ML equations easier to solve. Recall that $\theta_i^l=1$ if the degree of node $i$ is equal to the mean degree in layer $l$. 
For $(i,j) \in \Ee$, we write
\begin{align*}
	\theta_i^1\theta_j^1 &= 1+\varepsilon_{ij}^1 \,, \\
	\theta_i^2\theta_j^2 &= 1+\varepsilon_{ij}^2 \,.
\end{align*}
{If the degree distribution is narrow, such that all node degrees are close to the mean degree, then} 
$\varepsilon_{ij}^1$ and $\varepsilon_{ij}^2$ are small parameters (which can be either positive or negative). {In this case, a first-order Taylor expansion yields}
\begin{equation}\label{eqn:epsApprox1}
	\sqrt{\theta_i^1\theta_j^1\theta_i^2\theta_j^2} 
		=\sqrt{(1+\varepsilon_{ij}^1)(1+\varepsilon_{ij}^2)} 
		\approx 1 + \frac{\varepsilon_{ij}^1+\varepsilon_{ij}^2}{2}\,.
\end{equation}
We also calculate
\begin{align}\label{eqn:epsApprox2}
	\sqrt{\frac{\thi^1\thj^1}{\thi^2\thj^2}} =\sqrt{\frac{1+\varepsilon_{ij}^1}{1+\varepsilon_{ij}^2}} 
		\approx 1+\frac{\varepsilon_{ij}^1-\varepsilon_{ij}^2}{2}\,
\end{align} 
and 
\begin{equation}\label{eqn:epsApprox3}
	\sqrt{\frac{\thi^2\thj^2}{\thi^1\thj^1}} \approx 1+\frac{\varepsilon_{ij}^2-\varepsilon_{ij}^1}{2}\,. 
\end{equation}

Using the approximations \eqref{eqn:epsApprox1}--\eqref{eqn:epsApprox3}, we expand the first derivatives of $\mathcal{L}$ to first order in $\varepsilon_{ij}^1$ and $\varepsilon_{ij}^2$. (See \cite{pamfil2018thesis} for details.) This calculation yields the following system of equations: 
\begin{widetext}
\begin{align}
	&\frac{e_{rs}^{10}}{p_{rs}^1-q_{rs}}-\frac{e_{rs}^{00}}{1-p_{rs}^1-p_{rs}^2+q_{rs}} +\frac{g_{rs}^{10}}{2}\frac{q_{rs}}{(p_{rs}^1-q_{rs})^2}-f_{rs}^1\frac{1-p_{rs}^2+q_{rs}/2}{(1-p_{rs}^1-p_{rs}^2+q_{rs})^2}-f_{rs}^2\frac{p_{rs}^2-q_{rs}/2}{(1-p_{rs}^1-p_{rs}^2+q_{rs})^2} &= 0 \,, \notag \\
	&\frac{e_{rs}^{01}}{p_{rs}^2-q_{rs}}-\frac{e_{rs}^{00}}{1-p_{rs}^1-p_{rs}^2+q_{rs}} +\frac{g_{rs}^{01}}{2}\frac{q_{rs}}{(p_{rs}^2-q_{rs})^2}-f_{rs}^1\frac{p_{rs}^1-q_{rs}/2}{(1-p_{rs}^1-p_{rs}^2+q_{rs})^2}-f_{rs}^2\frac{1-p_{rs}^1+q_{rs}/2}{(1-p_{rs}^1-p_{rs}^2+q_{rs})^2} &= 0 \,, \notag \\
	&\frac{e_{rs}^{11}}{q_{rs}}-\frac{e_{rs}^{10}}{p_{rs}^1-q_{rs}}-\frac{e_{rs}^{01}}{p_{rs}^2-q_{rs}}+\frac{e_{rs}^{00}}{1-p_{rs}^1-p_{rs}^2+q_{rs}} \notag \\
	&\qquad\qquad\qquad\qquad\qquad  \quad - \frac{g_{rs}^{10}}{2}\frac{p_{rs}^1}{(p_{rs}^1-q_{rs})^2}-\frac{g_{rs}^{01}}{2}\frac{p_{rs}^2}{(p_{rs}^2-q_{rs})^2}
+\frac{1}{2}\frac{f_{rs}^1(1+p_{rs}^1-p_{rs}^2)+f_{rs}^2(1-p_{rs}^1+p_{rs}^2)}{(1-p_{rs}^1-p_{rs}^2+q_{rs})^2} &=0 \label{eqn:BernoulliDC1-3} \,.
\end{align}
\end{widetext}
In these equations, we defined $e_{rs}^{11}$, $e_{rs}^{10}$, $e_{rs}^{01}$, and $e_{rs}^{00}$ as in Sec. \ref{subsec:SBM}. 
Additionally, we set
\begin{align*}
	g_{rs}^{10}&=\sum_{(i,j) \in \Ee}{A_{ij}^1(1-A_{ij}^2)(\varepsilon_{ij}^2-\varepsilon_{ij}^1)}\delta(g_i,r)\delta(g_j,s) \,, \\
	g_{rs}^{01}&=\sum_{(i,j) \in \Ee}{(1-A_{ij}^1)A_{ij}^2(\varepsilon_{ij}^1-\varepsilon_{ij}^2)}\delta(g_i,r)\delta(g_j,s) \,, \\
	f_{rs}^1&=\sum_{(i,j) \in \Ee}{(1-A_{ij}^1)(1-A_{ij}^2)\varepsilon_{ij}^1}\delta(g_i,r)\delta(g_j,s) \,, \\
	f_{rs}^2&=\sum_{(i,j) \in \Ee}{(1-A_{ij}^1)(1-A_{ij}^2)\varepsilon_{ij}^2}\delta(g_i,r)\delta(g_j,s) \,.
\end{align*}
We can efficiently calculate all of these quantities from the matrices $\mat{A}^1$ and $\mat{A}^2$. 

The system of equations \eqref{eqn:BernoulliDC1-3} reduces to the analogous equations for correlated SBMs if we ignore all of the terms that depend on $\varepsilon_{ij}^l$ (i.e., the terms that correspond to perturbations of the degrees from their mean values). The zeroth-order solution that we obtain from ignoring these terms provides a good initialization of a numerical algorithm to solve \eqref{eqn:BernoulliDC1-3} for the parameters $p_{rs}^1$, $p_{rs}^2$, and $q_{rs}$.
In practice, when using correlated DCSBMs for edge prediction (see Sec.~\ref{sec:linkPrediction}), we find that using a first-order approximation to determine $p_{rs}^1$, $p_{rs}^2$, and $q_{rs}$ gives results that are almost identical to those from the zeroth-order approximation [see Eqns. \eqref{eqn:MLESBMBernoullip1}--\eqref{eqn:MLESBMBernoulliq}]. {Consequently, we suggest using these zeroth-order approximation for edge-prediction applications, given that they are straightforward to calculate and have negligible impact on the quality of the results.} {For large networks, we also obtain a noticeable improvement in calculation speed when using these approximations.}

{In App.~\ref{app:approximation}, we compare the parameters that we estimate using the approximate system of equations \eqref{eqn:BernoulliDC1-3} with those from the log-likelihood \eqref{eqn:DCSBMlogL}. As expected, the quality of the approximation depends on the shape of the degree distribution, with larger discrepancies between the two approaches for broader degree distributions.}


\subsubsection{Correlation values}

For the SBMs without degree correction from Sec.~\ref{subsec:SBM}, node pairs $(i,j)$ from a given edge bundle $(r,s)$ have the same Pearson correlation $\rho_{rs}$. This no longer holds for degree-corrected models. Instead, each node pair $(i,j)$ has its own correlation value
\begin{align}\label{eqn:rhoDC}
	\varrho_{ij} &= \frac{\E[A_{ij}^1A_{ij}^2]-\E[A_{ij}^1]\E[A_{ij}^2]}{\sigma[A_{ij}^1]\sigma[A_{ij}^2]} \\
&=  \frac{\sqrt{\thi^1\thj^1\thi^2\thj^2}q_{rs}-\thi^1\thj^1\thi^2\thj^2p_{rs}^1p_{rs}^2}{\sqrt{\thi^1\thj^1p_{rs}^1(1-\thi^1\thj^1p_{rs}^1)\thi^2\thj^2p_{rs}^2(1-\thi^2\thj^2p_{rs}^2)}}\,. \notag
\end{align}

As in our earlier expansions of the ML equations, we approximate $\varrho_{ij}$ to first order in $\varepsilon_{ij}^1$ and $\varepsilon_{ij}^2$. We obtain
\begin{align}
	\varrho_{ij}
		\approx \rho_{rs}+\rho_{rs}&\Biggl( \frac{\varepsilon_{ij}^1}{2}\frac{p_{rs}^1}{1-p_{rs}^1} +
\frac{\varepsilon_{ij}^2}{2}\frac{p_{rs}^2}{1-p_{rs}^2} \notag \\
		&- \frac{\varepsilon_{ij}^1 +\varepsilon_{ij}^2}{2}\frac{p_{rs}^1p_{rs}^2}{q_{rs}-p_{rs}^1p_{rs}^2} \Biggr) \label{eqn:rhoDCapprox} \,.
\end{align}
Ignoring terms that depend on $\varepsilon_{ij}^1$ and $\varepsilon_{ij}^2$ (i.e., terms that correspond to perturbations of the degrees from their mean values), we obtain $\varrho_{ij} \approx \rho_{rs}$. This approximation works especially well when $p_{rs}^1$ and $p_{rs}^2$ are also small, such that their respective network layers are sparse. 

The case $q_{rs}=p_{rs}^1p_{rs}^2$ requires separate consideration \footnote{Situations when one or more of the denominators in Eqn. \eqref{eqn:rhoDCapprox} are $\mathcal{O}(\varepsilon)$ also require separate derivations. However, we do not treat these exceptional cases in the present paper.} to avoid dividing by $0$.
First-order approximations in $\varepsilon_{ij}^1$ and $\varepsilon_{ij}^2$ for this case give
\begin{align}
	\varrho_{ij} &\approx -\frac{\varepsilon_{ij}^1+\varepsilon_{ij}^2}{2}\sqrt{\frac{p_{rs}^1}{1-p_{rs}^1}}\sqrt{\frac{p_{rs}^2}{1-p_{rs}^2}}\,.
\end{align}
In particular, the zeroth-order solution gives $\varrho \approx 0$, in agreement with the SBM without degree correction from Sec. \ref{subsec:SBM}.


\section{Edge Prediction}\label{sec:linkPrediction}

The aim of edge prediction (also called ``link prediction'') in networks is to infer likely missing edges and/or spurious edges \cite{lu2011}. Edge prediction is useful for filling in incomplete data sets, such as protein-interaction networks (in which edges are often established as a result of costly experiments) \cite{guimera2009} or terrorist-association networks (which are typically constructed based on partial knowledge) \cite{clauset2008}. In the context of bipartite user--item networks, edge-prediction techniques provide candidates for personalized recommendations.

{
One can perform edge prediction in either a supervised or an unsupervised fashion. We briefly discuss each of these types of approaches.}

{
Supervised methods rely on models that learn how a specified set of features relates to the presence or absence of edges. Existing methods that take advantage of multilayer structure for edge prediction typically do so through the specification of multilayer features. These include aggregations of monolayer features \cite{pujari2015}, as well as path-based \cite{jalili2017} and neighborhood-based \cite{hristova2016,mandal2018} features that consider multiple layers. 
Although many of these features depend indirectly on the similarity between different layers, none of these methods quantify the level of correlation or use it for edge prediction. 
}

{
With unsupervised methods, one obtains a ranking of node pairs such that edges are more likely among higher-ranked pairs. Common approaches include ones that are based on probabilistic models and ones that are based on similarity indices (like the Jaccard index or the Adamic--Adar index) \cite{lu2011}. An example of the former for multilayer networks is a method that maps each network layer independently to a hyperbolic space and then uses the hyperbolic distance between nodes in one layer to predict edges in another layer \cite{kleineberg2016}. This work used node-centric notions of correlation and thereby complements the edge-centric perspective of our work. 
}
{Methods that rely on similarity indices include those that first generate latent states (i.e., so-called ``embeddings'') for the nodes and then rank node pairs according to the similarities of these embedding vectors \cite{matsuno2018,pio2020}. Tillman et al. \cite{tillman2020} used layer-level correlations to combine monolayer similarity indices into a single score. There are also approaches that implicitly take advantage of similarities across layers, such as by extracting common higher-order structures (specifically, subgraphs with three or more nodes) and looking for patterns that differ by exactly one edge \cite{coscia2020}.}

{
Variants of SBMs are popular choices for unsupervised edge-prediction methods \cite{guimera2009,guimera2013,clauset2008,peixoto2017bayesian}, including in multilayer settings \cite{debacco2017,valles2016multilayer}. As an example, in a monolayer degree-corrected Bernoulli SBM, the probability that two nodes, $i$ and $j$, are adjacent according to the model is $\Pp(A_{ij}=1) = \theta_i\theta_jp_{g_ig_j}$. The pairs $(i,j)$ for which these probabilities are relatively large but which are not adjacent (i.e., with $A_{ij}=0$) in the actual network of interest produce a list of likely candidates for missing edges. Similarly, pairs $(i,j)$ for which these probabilities are small but which are adjacent (i.e., with $A_{ij}=1$) in the actual network may be spurious edges.
}


\subsection{Edge Prediction Using Correlated Models}

There have been several recent attempts to perform edge prediction in multilayer networks \cite{debacco2017,valles2016multilayer,tarres2018}. All of these methods use multilayer information to infer mesoscale structures in networks, but then they perform edge prediction independently in each layer, conditioned on the inferred mesoscale structure {and any other model parameters}. 
In particular, when using one of these approaches, {observing that two nodes are adjacent} in one layer has no bearing on their probability to be adjacent in another layer. We aim to use our correlated models to overcome this limitation.

As in our prior discussions, consider a network with two layers with intralayer adjacency matrices $\mat{A}^1$ and $\mat{A}^2$, and let $\vec{g}$ be the shared block structure of these layers. Our goal is to predict edges in the second layer, conditioned on the adjacency structure of the first layer. For each node pair $(i,j) \in \Ee$, the key quantities to calculate are the probabilities $\Pp(A_{ij}^2=1|A_{ij}^1=1)$ and $\Pp(A_{ij}^2=1|A_{ij}^1=0)$ for $i$ and $j$ to be adjacent in the second layer, conditioned on them either being adjacent or non-adjacent in the first layer. For example, using the correlated Bernoulli SBM from Sec.~\ref{subsec:SBM} (which has no degree correction), we have
\begin{align}
	\Pp(A_{ij}^2=1|A_{ij}^1=1) &= 
\frac{q_{g_ig_j}}{p_{g_ig_j}^1} \,, \label{eqn:edgePredCond1} \\
	\Pp(A_{ij}^2=1|A_{ij}^1=0) &=
\frac{p_{g_ig_j}^2-q_{g_ig_j}}{1-p^1_{g_ig_j}} \label{eqn:edgePredCond2} \,.
\end{align}
This set of probabilities is the same across all node pairs $(i,j)$ from the same edge bundle $(r,s)$. Now suppose that we have a positive correlation in this edge bundle, so $\rho_{rs}>0$. From the definition of the Pearson correlation, it follows that $q_{rs}>p_{rs}^1p_{rs}^2$. We then find that $\Pp(A_{ij}^2=1|A_{ij}^1=1)>p_{rs}^2$ and $\Pp(A_{ij}^2=1|A_{ij}^1=0)<p_{rs}^2$, whereas using a monolayer SBM would entail that $\Pp(A_{ij}^2=1)=p_{rs}^2$. Therefore, the effect of incorporating correlations into our edge-prediction model when these correlations are positive is (1) to increase the probability that nodes $i$ and $j$ are adjacent in the second layer when the corresponding edge also exists in the first layer and (2) to decrease this probability when the corresponding edge is absent from the first layer. The effect is reversed for negative correlations.

\begin{table*}[htb!]
\caption{Edge-prediction probabilities for various correlated multilayer and monolayer network models.}
\label{table:edgePredictionSummary}
\renewcommand{\arraystretch}{1.3}
\centering
\begin{tabular}{l|c|c}
\hline\hline
Model & $\Pp(A_{ij}^2=1|A_{ij}^1=1)$ & $\Pp(A_{ij}^2=1|A_{ij}^1=0)$ \\ \hline
Corr. ER      & $q/p_1$           & $(p_2-q)/(1-p_1)$ \\
Corr. SBM     & $q_{rs}/p_{rs}^1$ & $(p_{rs}^2-q_{rs})/(1-p_{rs}^1)$ \\
Corr. CM      & $\left.\sqrt{\theta_i^1\theta_j^1\theta_i^2\theta_j^2}q\right/(\theta_i^1\theta_j^1p_1)$ 
  & $(\theta_i^2\theta_j^2p_2-\left.\sqrt{\theta_i^1\theta_j^1\theta_i^2\theta_j^2}q)\right/
  (1-\theta_i^1\theta_j^1p_1)$ \\
Corr. DCSBM   & $\left.\sqrt{\theta_i^1\theta_j^1\theta_i^2\theta_j^2}q_{rs}\right/(\theta_i^1\theta_j^1p_{rs}^1)$ & $(\theta_i^2\theta_j^2p_{rs}^2-\left.\sqrt{\theta_i^1\theta_j^1\theta_i^2\theta_j^2}q_{rs})\right/
  (1-\theta_i^1\theta_j^1p_{rs}^1)$ \\ \hline
  SBM         & \multicolumn{2}{c}{$p_{rs}^2$} \\
  DCSBM       & \multicolumn{2}{c}{$\theta_i^2\theta_j^2p_{rs}^2$} \\
\hline\hline
\end{tabular}
\end{table*}

In Table \ref{table:edgePredictionSummary}, we summarize the two key probabilities \eqref{eqn:edgePredCond1}--\eqref{eqn:edgePredCond2} for four different correlated models, alongside the probabilities $\Pp(A_{ij}^2=1)$ for monolayer SBMs and DCSBMs. We include a correlated ``configuration model" (CM) \cite{fosdick2018}, which is a special case of the degree-corrected SBM from Sec.~\ref{subsec:DCSBM} when there is only one block. (Alternatively, one can think of correlated CMs as extensions of correlated ER models that incorporate degree correction.) We use all of these models for edge prediction in synthetic networks in Sec.~\ref{subsec:edgePredictionSynthetic} and in consumer--product networks in Sec.~\ref{subsec:shopping}. In particular, the two monolayer models are baselines that we hope to outperform using our correlated models.


\subsection{Tests on Synthetic Networks}\label{subsec:edgePredictionSynthetic}

We use $K$-fold \footnote{Note that this $K$ is different from the one that denotes the number of blocks elsewhere in this paper.} cross-validation to assess the performance of the models from Table \ref{table:edgePredictionSummary} on the edge-prediction task. In machine learning, this is an effective way to measure predictive performance \cite{lu2011}. 
After partitioning a given data set into $K$ parts, one 
fits a model to $K-1$ of these subsets and uses it to make predictions on the remaining (i.e., ``holdout") set. 
One uses each subset once as a holdout, so
one does this process $K$ times in total. For our problem, we perform $5$-fold cross-validation (which is a standard choice in the machine-learning literature) by splitting the data in the second layer of a given network into $5$ subsets. Effectively, this consists of hiding $20\%$ of the entries of the adjacency matrix $\mat{A}^2$, such that we do not know whether they are edges or not. 
We then train a model on $100\%$ of the entries of $\mat{A}^1$ and $80\%$ of the entries of $\mat{A}^2$, and we use it to make predictions about the $20\%$ holdout data from $\mat{A}^2$. We do this $5$ times to cover each choice of holdout data.

A common way to assess the performance of a binary classification model (i.e., a model that assigns one of two possible values to test data) is by using a receiver operating characteristic (ROC) curve. An ROC curve plots the true-positive rate (TPR) of a classifier versus the false-positive rate (FPR) for various choices of a threshold. Many models --- including those that are used for edge prediction in networks --- make probabilistic predictions, so specifying a threshold is necessary to convert these into binary predictions. Lowering the threshold increases both the TPR and the FPR. A model has predictive power if the former grows faster than the latter, such that the entire ROC curve lies above the diagonal line $\mbox{TPR}=\mbox{FPR}$, which gives the performance of a random classifier. As a single summary measure of a model's predictive performance, it is common to report the area under an ROC curve (AUC). Larger AUC values are better, with a value of $\mathrm{AUC}=0.5$ indicating equal success as random guessing and $\mathrm{AUC}=1$ corresponding to perfect prediction. Even in the latter case, one still needs to determine a choice of threshold that completely separates true positives from false positives.

{The AUC is not the only possible quantity to assess edge-prediction performance, although it is very common \cite{debacco2017, ghasemian2019}. In many networks, the number of edges is much smaller than the number of non-edges (i.e., pairs of nodes that are not adjacent). In situations with such an imbalance, the area under the precision--recall (PR) curve is more sensitive than the AUC to variations in model performance. Nevertheless, we use the AUC because it has an intuitive interpretation (specifically, as the probability that the underlying model ranks a true positive above a true negative) that allows us to establish the result in App.~\ref{app:AUCproof}.} {Additionally, the conclusions that we draw in Sec.~\ref{subsec:shopping} about the performance of different models do not change significantly if we use PR curves rather than ROC curves.}

\begin{figure*}[bht]
\centering
\subfloat[Models without mesoscale structure]{\includegraphics[width=0.35\textwidth]{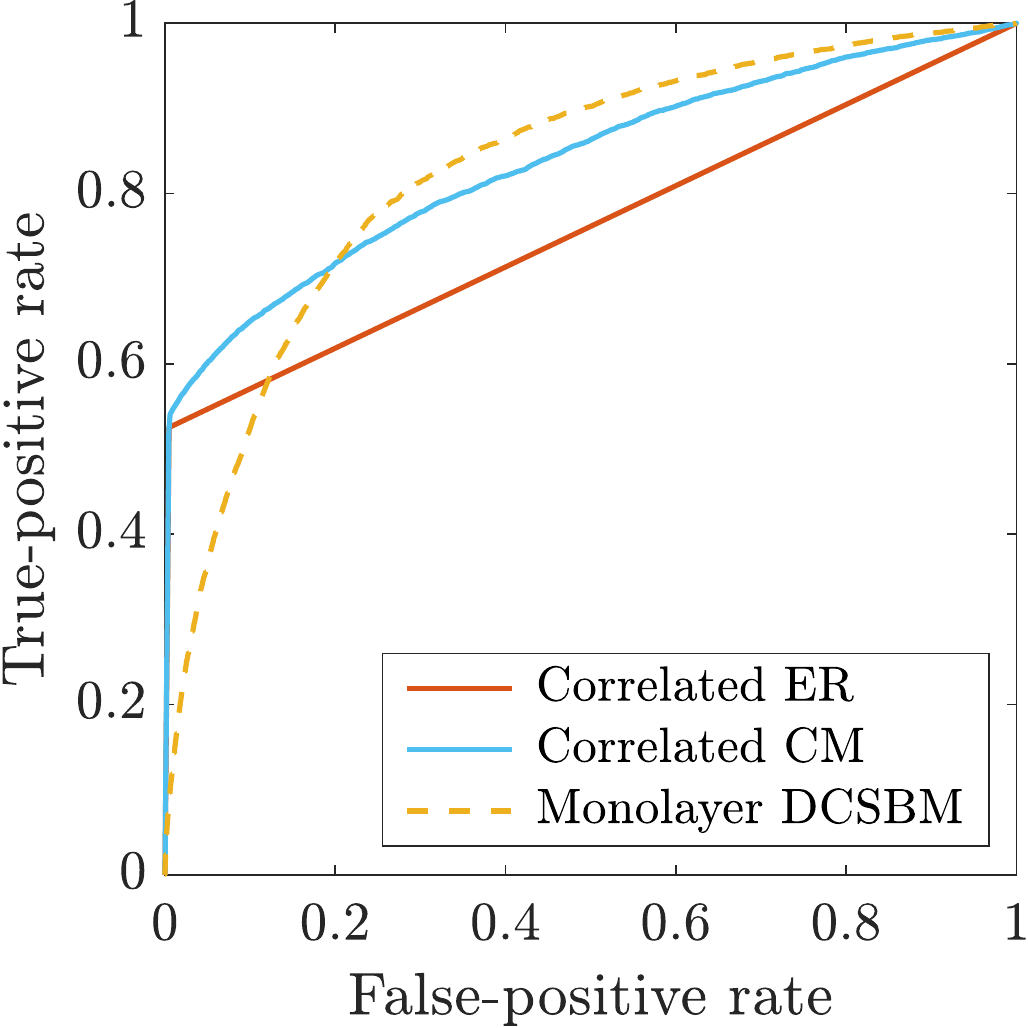}} 
\hspace{6em}
\subfloat[Models with mesoscale structure]{\includegraphics[width=0.35\textwidth]{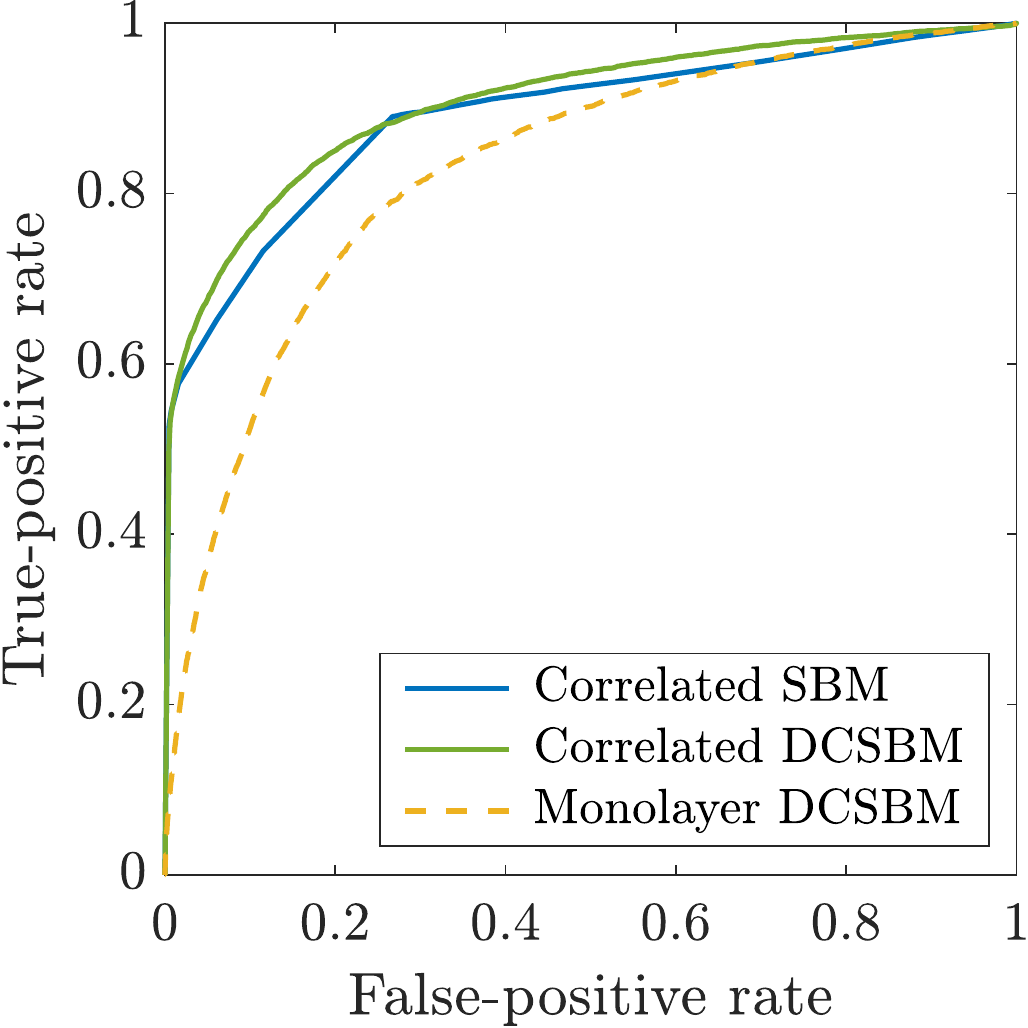}} 
\caption{ROC curves from a $5$-fold cross-validation for a network that we sample from the {\sc CorrDCSBM} benchmark with community-mixing parameter $\mu=0.3$ and correlation $\rho=0.5$. 
(a) Correlated models that do not incorporate any mesoscale structure compared to a monolayer DCSBM baseline. The baseline gives $\mathrm{AUC} \approx 0.83$, whereas the AUC values for the two correlated models are approximately $0.76$ (correlated ER) and $0.83$ (correlated CM). (b) Correlated models that incorporate mesoscale structure compared to the same monolayer DCSBM baseline. The baseline again gives $\mathrm{AUC} \approx 0.83$, and the AUC values for the two correlated models are approximately $0.89$ (correlated SBM) and $0.91$ (correlated DCSBM).
}
\label{fig:edgePredROC}
\end{figure*}

\begin{figure*}[bth!]
\centering
\subfloat[$\mu=0.3$, $\rho \leq 0$]{\includegraphics[width=0.48\textwidth]{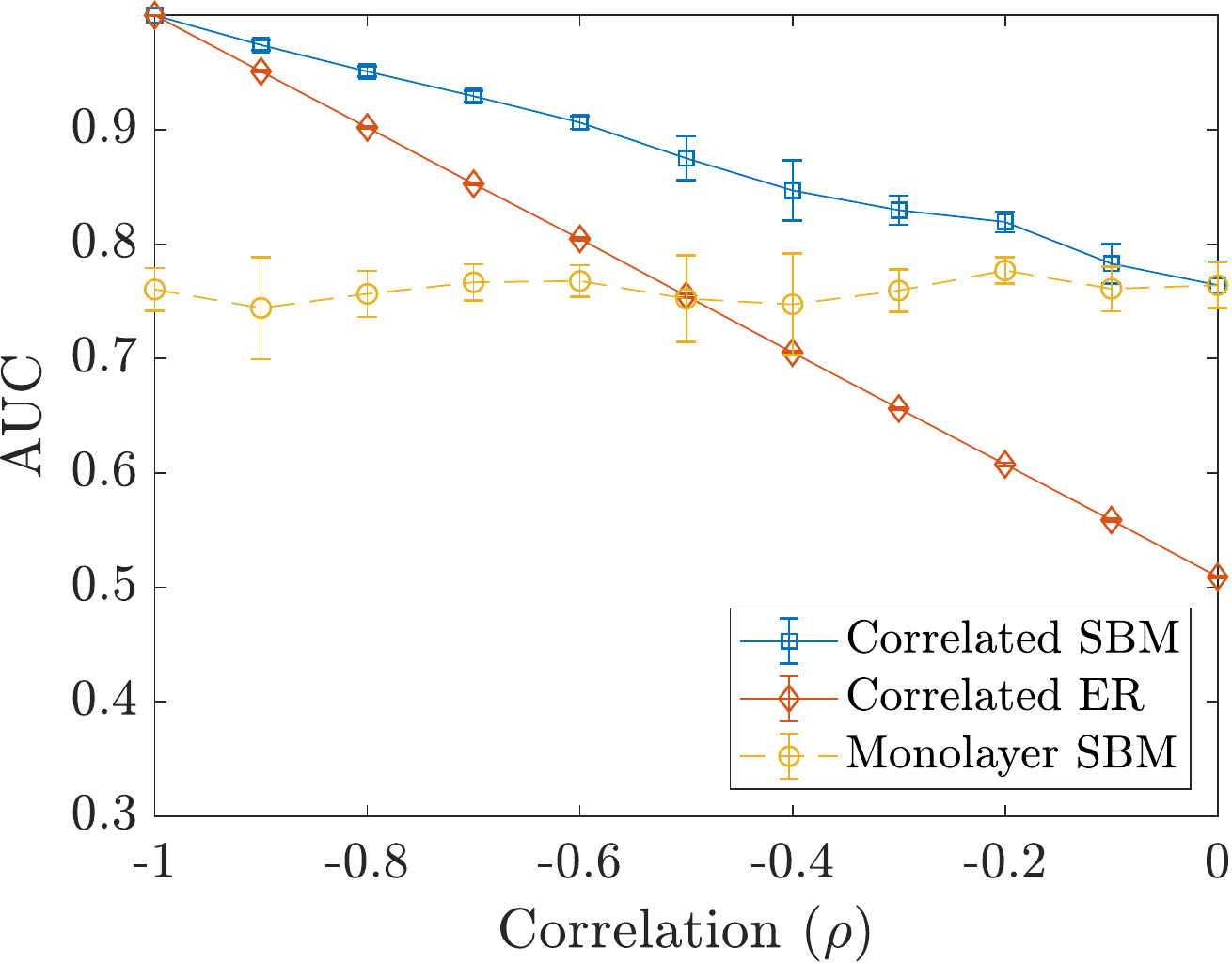}} 
\hfill
\subfloat[$\mu=0.3$, $\rho \geq 0$]{\includegraphics[width=0.48\textwidth]{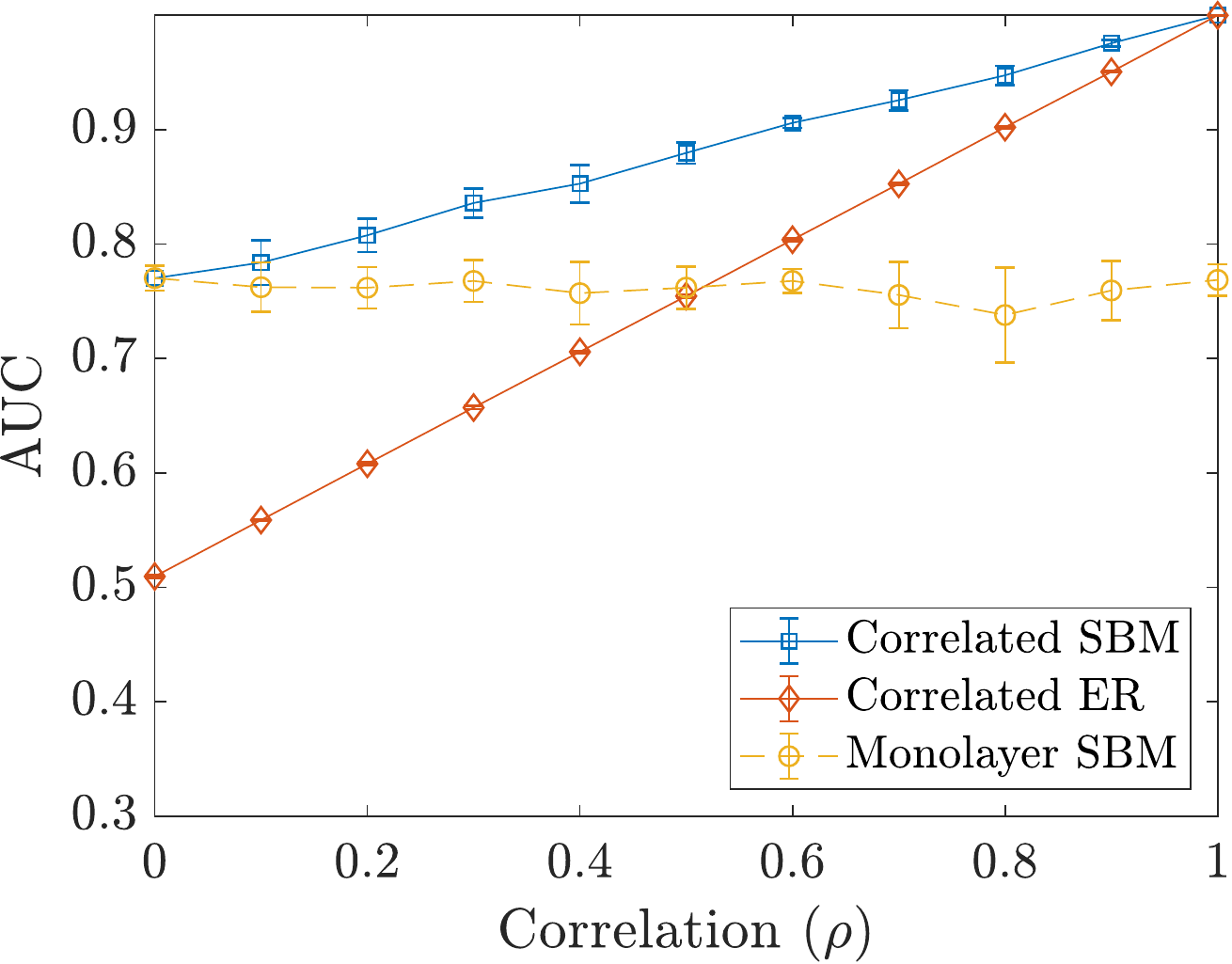}} \\
\subfloat[$\mu=0.8$, $\rho \leq 0$]{\includegraphics[width=0.48\textwidth]{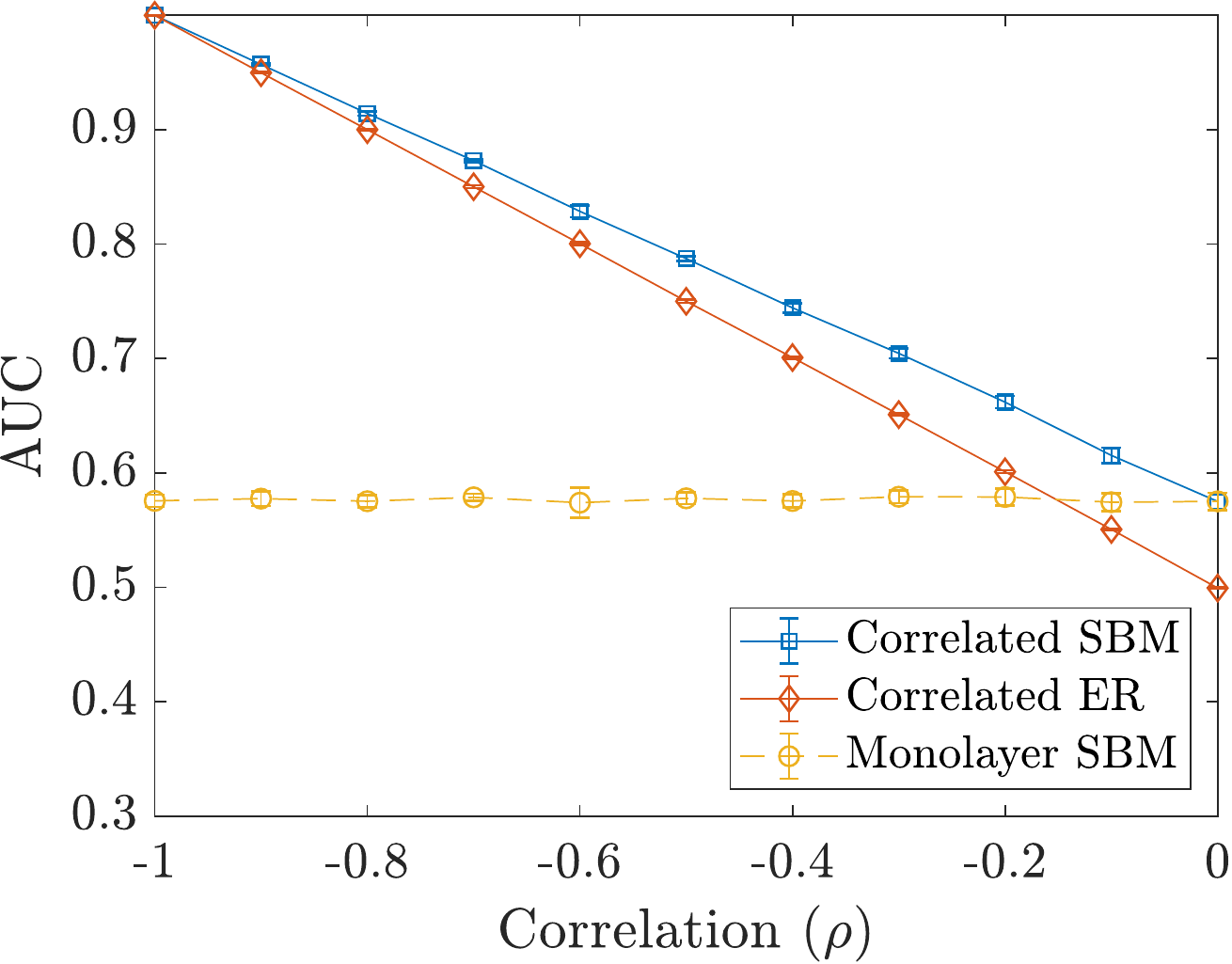}} 
\hfill
\subfloat[$\mu=0.8$, $\rho \geq 0$]{\includegraphics[width=0.48\textwidth]{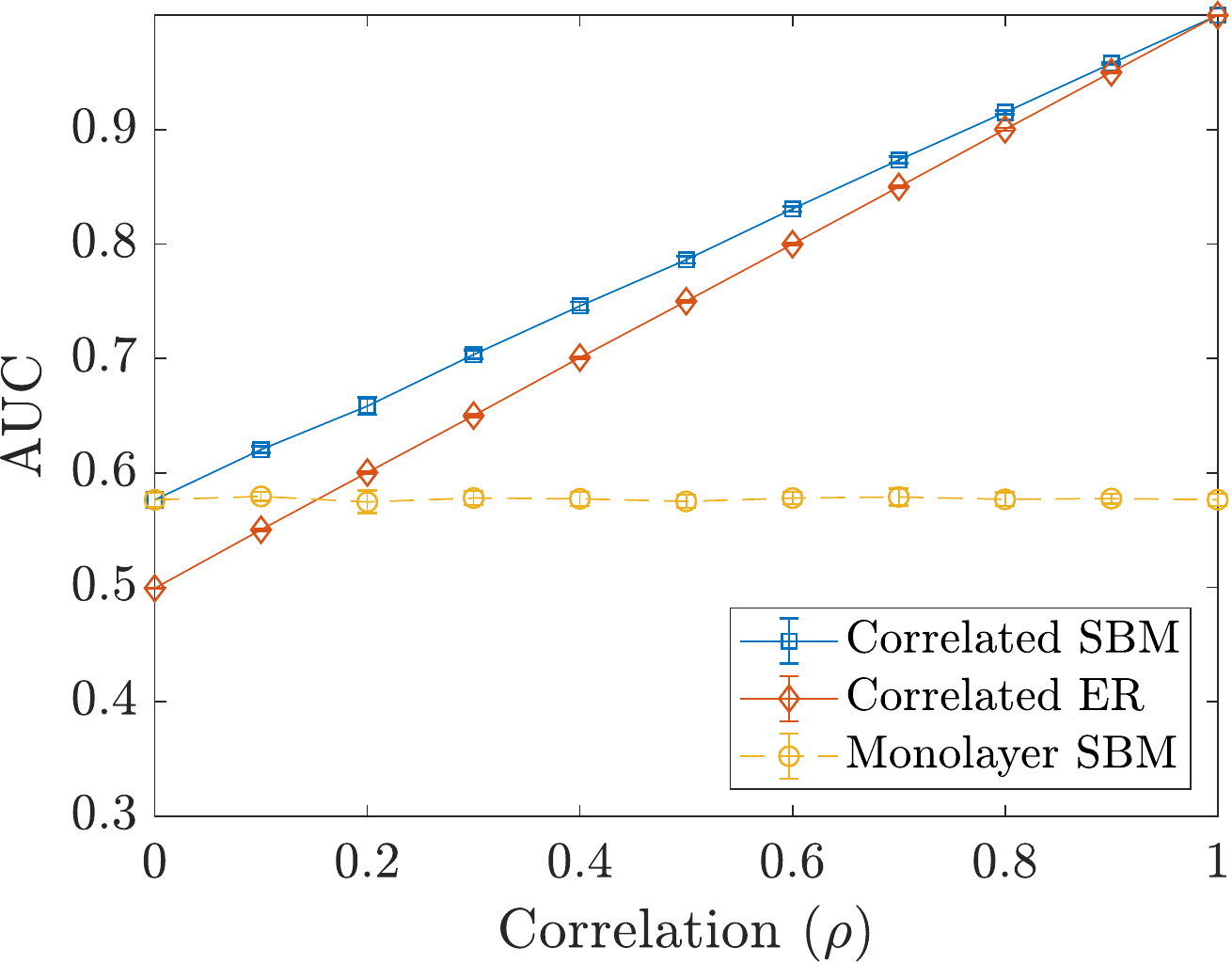}} 
\caption{Edge-prediction results on synthetic networks from the {\sc CorrSBM} benchmark with (left) $\rho \leq 0$ and (right) $\rho \geq 0$ using two choices of the community-mixing parameter $\mu$. We use $\mu=0.3$ in the top row and $\mu=0.8$ in the bottom row. In all plots, along the horizontal axis, we vary the correlation $\rho$ that we use to generate network instances. On the vertical axis, we indicate the AUC for $5$-fold cross-validation using a monolayer SBM (dashed curves) and correlated SBM and ER models (solid curves). Each data point is a mean across $10$ trials, and the error bars correspond to one standard deviation from that mean. As expected, the AUC does not change with $\rho$ for the monolayer model, but it increases with $|\rho|$ for the two correlated models. For progressively larger $\mu$, for which the sampled networks have progressively weaker mesoscale structure, there is a smaller performance gap between correlated ER models and correlated SBMs.
  }
\label{fig:edgePred}
\end{figure*}

We now describe how to generate synthetic network benchmarks that are suitable for testing the models in Table \ref{table:edgePredictionSummary}. We construct these networks so that they have two tunable parameters: the Pearson correlation $\rho \in [-1,1]$ and a community-mixing parameter $\mu \in [0,1]$ that controls the strength of the planted mesoscale structure. See Bazzi et al. \cite{bazzi2016generative} for more details about the definition of $\mu$. One can also explicitly control the degree distribution, such as by including a parameter $\eta_k$ for the slope of a truncated power law (e.g., as used in \cite{bazzi2016generative} to sample a degree sequence in each layer). For the experiments in this section, we fix $\eta_k=-2$ and use a minimum degree of $k_\mathrm{min}=10$ and a maximum degree of $k_\mathrm{max}=50$. It would be interesting to explore the performance gap between degree-corrected models and models without degree correction as one varies $\eta_k$, $k_\mathrm{min}$, and $k_\mathrm{max}$, although we do not do so in the present paper. Lastly, for our numerical experiments in this section, we use networks with $N=2000$ nodes in each layer and $n_c=5$ communities, with community sizes sampled from a flat Dirichlet distribution (i.e., one with $\theta=1$ in the notation of \cite{bazzi2016generative}).

\begin{figure*}[tb!]
\centering
\subfloat[$\mu=0.3$]{\includegraphics[width=0.48\textwidth]{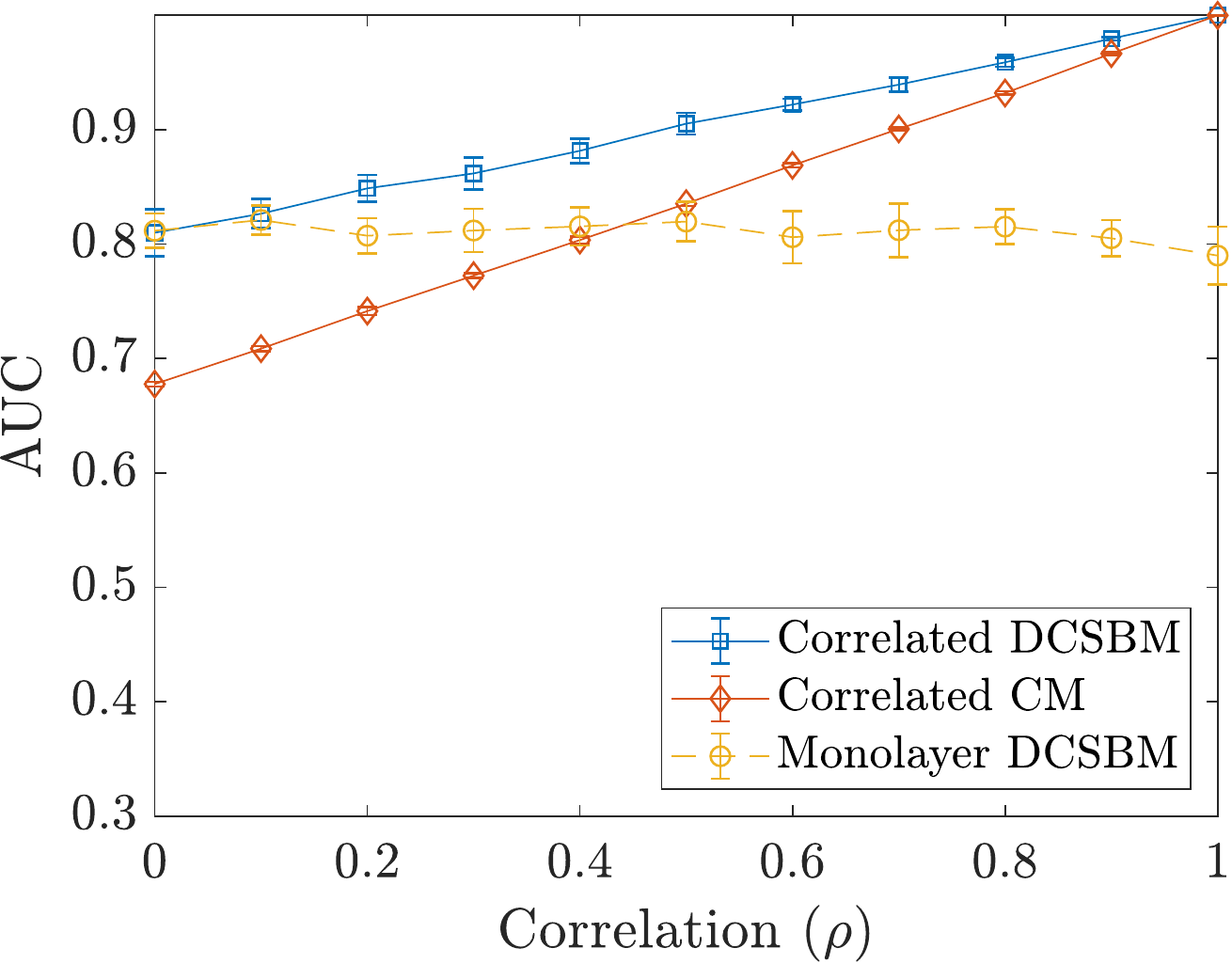}} 
\hfill
\subfloat[$\mu=0.8$]{\includegraphics[width=0.48\textwidth]{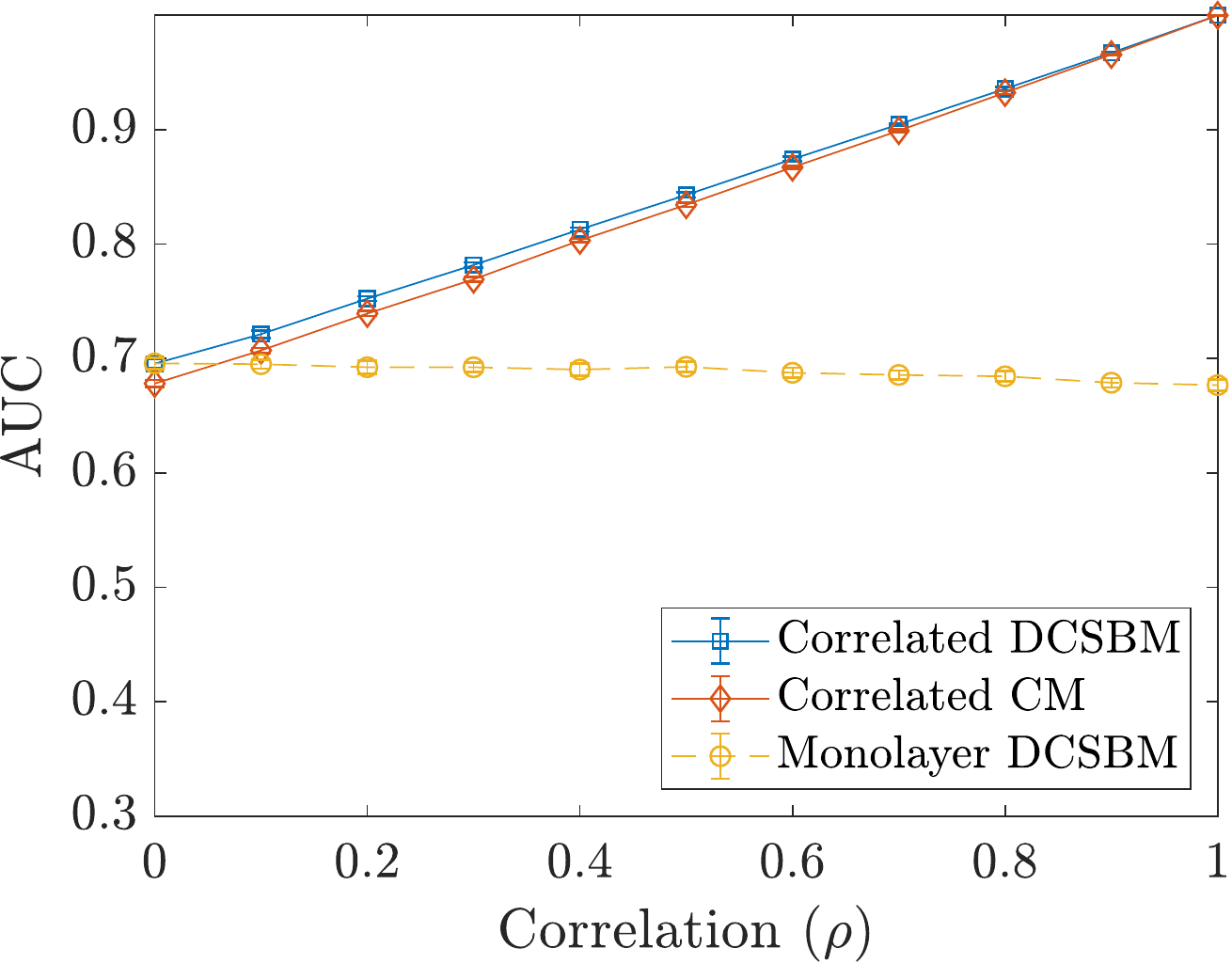}} 
\caption{Edge-prediction results on synthetic networks from the {\sc CorrDCSBM} benchmark with correlation $\rho \geq 0$ and community-mixing parameters of  (a) $\mu=0.3$ and (b) $\mu=0.8$. In both panels, along the horizontal axis, we vary the correlation $\rho$ that we use to generate network instances. On the vertical axis, we indicate the AUC for $5$-fold cross-validation using a monolayer DCSBM (dashed curves) and the correlated DCSBM and CM (solid curves).
Each data point is a mean across $10$ trials, and the error bars correspond to one standard deviation from that mean. As expected, the AUC is roughly independent of $\rho$ for the monolayer model, but it increases with $\rho$ for the two correlated models. As we increase $\mu$, such that the sampled networks have progressively weaker mesoscale structure, we observe a substantial narrowing of the performance gap between the correlated CMs and the correlated DCSBMs.}
\label{fig:edgePredDC}
\end{figure*}

We examine two versions, which we call {\sc CorrSBM} and {\sc CorrDCSBM}, of a correlated benchmark that is parametrized by the correlation $\rho$ and the community-mixing parameter $\mu$. For both versions of the benchmark, we generate the (undirected and unipartite) adjacency matrix $\mat{A}^1$ of the first layer in the same way. Specifically, given $\mu$ and degree-distribution parameters $\eta_k$, $k_\mathrm{min}$, and $k_\mathrm{max}$, we use the code from \cite{benchmark} to generate $\mat{A}^1$ and its associated block structure $\vec{g}$. We fit a monolayer model --- either an SBM or a DCSBM, depending on the selected version of the benchmark ---  to $\mat{A}^1$ to obtain the marginal edge propensities $\mat{p}^1$ for the first layer. We then choose $\mat{p}^2$ in one of two ways. For $\rho \in [0,1]$, we set $\mat{p}^2=\mat{p}^1$, which ensures that we can generate networks with correlations that cover the entire range from $0$ to $1$. For $\rho \in [-1,0]$, we set $\mat{p}^2=\mat{\mathbbm{1}}-\mat{p}^1$, where $\mat{\mathbbm{1}}$ is a
matrix with all entries equal to $1$; this ensures that we can generate networks with correlations that cover the entire range from $-1$ to $0$.
Given $\mat{p}^1$, $\mat{p}^2$, and $\rho$, we then determine $\mat{q}$ using either the correlated SBM of Sec.~\ref{subsec:SBM} or the correlated DCSBM of Sec.~\ref{subsec:DCSBM}. For the {\sc CorrDCSBM} benchmark with $\rho \geq 0$, we set the normalized degrees $\theta_i^2$ to be equal to the corresponding quantities $\theta_i^1$ from the first layer. Again, this choice ensures that we can generate networks all the way to $\rho=1$. (We also implemented a version of this benchmark that samples degrees independently in the second layer, and we found qualitatively similar results.) The final step consists of generating $\mat{A}^2$ given $\mat{A}^1$, the propensities $\mat{p}^2$ and $\mat{q}$, and (for the {\sc CorrDCSBM} benchmark only) the normalized degree sequences $\vec{\theta}^1$ and $\vec{\theta}^2$ for both layers. To perform this step, we first compute edge probabilities using either of the correlated models from Secs.~\ref{subsec:SBM} and \ref{subsec:DCSBM}, and we then generate edges independently according to these probabilities.

We now present our results for the two variants of the benchmark. In Fig.~\ref{fig:edgePredROC}, we show sample ROC curves for one network that we create using the {\sc CorrDCSBM} benchmark with $\mu=0.3$ and $\rho=0.5$. We compare the performance of our correlated models with a monolayer DCSBM baseline, which performs edge prediction using only information from the second network layer. Two of the correlated models outperform this baseline, and the correlated CM performs comparably well (i.e., it has a similar AUC).

In Fig.~\ref{fig:edgePred}, we show results for the {\sc CorrSBM} benchmark for two choices of the community-mixing parameter $\mu$ and several values (both positive and negative) of the Pearson correlation $\rho$. As expected, the AUC values for monolayer SBMs are independent of $\rho$, whereas the predictive performance of correlated ER models and correlated SBMs improves as we increase $|\rho|$. In particular, when $|\rho|=1$, the two correlated models make perfect predictions. When $\rho=0$, the performance of the correlated ER model is indistinguishable from chance (because $\mathrm{AUC} \approx 0.5$), whereas correlated SBMs {have identical performance to} monolayer SBMs. The gap between the two correlated models is smaller for $\mu=0.8$ than for $\mu=0.3$, because the underlying block structure is weaker in the former case than in the latter. The AUC of the monolayer baseline is also smaller for $\mu=0.8$ than for $\mu=0.3$.

One striking feature in Fig.~\ref{fig:edgePred} is that all curves are approximately straight lines (to within sampling error). It makes sense that the performance of monolayer SBMs does not vary with $\rho$, as these models do not use any information from the other layer, but the linear dependence on $\rho$ of the other two curves is less intuitive.
For the correlated ER model, we can establish rigorously (see App.~\ref{app:AUCproof}) that the AUC is approximately equal to $(1+|\rho|)/2$ when $p_1 \approx p_2$ or when $p_1 \approx 1-p_2$. Given that the correlated SBM curves from Fig.~\ref{fig:edgePred} also exhibit a linear dependence on $\rho$, we believe that it is possible to establish similar results for correlated models that incorporate mesoscale structure. These results have practical importance, as they allow one to quickly estimate the additional benefits of using
correlated models instead of monolayer SBMs for edge prediction.  

In Fig.~\ref{fig:edgePredDC}, we show results for the {\sc CorrDCSBM} benchmark for two choices of the community-mixing parameter $\mu$ and nonnegative \footnote{The case $\rho \leq 0$ requires a different way of sampling normalized degrees $\theta_i^2$ of nodes in the second layer to be able to take $\rho$ all the way to $\rho=-1$. More precisely, when $\rho$ is close to $-1$, the second layer is dense if the first layer is sparse (and vice versa), so setting $\vec{\theta}^1=\vec{\theta}^2$ is inadequate.} values of the Pearson correlation $\rho$. As expected, when $\rho=0$, correlated DCSBMs perform similarly to monolayer DCSBMs. As in Fig.~\ref{fig:edgePred}, the performance of the monolayer model is roughly independent of $\rho$, whereas the two correlated models perform better as $\rho$ increases. The gap between the two correlated models narrows substantially as one increases $\mu$ from $0.3$ to $0.8$. 


\section{Applications}\label{sec:applications}

We discuss two applications of correlated multilayer-network models to the analysis of empirical networks. In Sec.~\ref{subsec:layerCorr}, we report pairwise layer correlations for several multiplex networks of different sizes. In Sec.~\ref{subsec:shopping}, we consider a temporal bipartite network of customers and products. Using an approach similar to that from Sec.~\ref{sec:linkPrediction}, we demonstrate that correlated multilayer models have a better edge-prediction performance than monolayer baselines.


\subsection{Layer Correlations in Empirical Networks}\label{subsec:layerCorr}

We now calculate pairwise layer correlations using the formula \eqref{eqn:rhoSBMBernoulliEff}. Recall that this expression gives the effective correlation between two layers, assuming that they have identical block structures (although their edge-propensity parameters can be different). Crucially, this calculation does \textit{not} require that one first determines the underlying block structure. In fact, as we demonstrated in Sec.~\ref{subsec:SBM}, the effective correlation of a correlated SBM recovers the correlation of a correlated ER graph, which is straightforward to compute. Accounting for node degrees, as we did in Sec.~\ref{subsec:DCSBM} for correlated DCSBMs, significantly increases the complexity of such a calculation. Additionally, as we showed in Sec.~\ref{subsec:DCSBM}, correlations using a degree-corrected model are rather similar to those that one obtains without degree correction.

In Table~\ref{tab:empiricalCorr}, we report the mean pairwise layer correlation for $9$ multiplex networks. (See App.~\ref{app:dataSets} for descriptions of these data sets.) To provide additional insight into these networks, we also report the two layers with the largest effective correlations. 

We make a few observations about some of the results in Table~\ref{tab:empiricalCorr}. For the \textit{C. elegans} connectome, the layers that correspond to two types of chemical synapses are highly correlated with each other, and their correlation to the layer of electrical synapses is comparatively lower. For the European Union (EU) air transportation network, the two most correlated layers are those that correspond to Scandinavian Airlines and Norwegian Air Shuttle flights; this is consistent with the findings in \cite{kao2017} (which were based on a different method for quantifying layer similarity). For the network of arXiv collaborations between network scientists, the two most similar categories are ``physics.data-an'' (which stands for ``Data Analysis, Statistics and Probability'') and ``cs.SI'' (which stands for ``Social and Information Networks''). One hypothesis is that these two labels are often used together in papers; such common usage results in a large edge overlap between the corresponding layers and hence in a large correlation value.

\begin{table*}[htb!]
\centering
\caption{Pairwise layer correlations in several multiplex networks.
}
\begin{tabular}{l|l|c|c|l}
\hline\hline
\multirow{2}{*}{Domain} & \multirow{2}{*}{Network} & Number & Mean & \multirow{2}{*}{Largest correlation (corresponding layers)} \\ 
& & of layers & correlation & \\ \hline
\multirow{3}{*}{Social} & 
CS Aarhus \cite{magnani2013} & 5 & $0.27$ & $0.45$ (``work'' and ``lunch'' layers) \\
& Lazega law firm \cite{lazega2001} & 3 & $0.39$ & $0.48$ (``advice'' and ``co-work'' layers) \\
& YouTube \cite{tang2009} & $5$ & $0.12$ & $0.20$ (``shared subscriptions'' and ``shared subscribers'') \\ \hline
\multirow{3}{*}{Biological}
& \textit{C. elegans} connectome \cite{chen2006} & 3 & $0.47$ & $0.85$ (``MonoSyn'' and ``PolySyn'' layers) \\
& \textit{P. falciparum} genes \cite{larremore2013} & 9 & $0.08$ & $0.25$ (``HVR7'' and ``HVR9'' layers) \\
& \textit{Homo sapiens} proteins \cite{stark2006} & 7 & $0.04$ & $0.29$ (``direct interaction'' and ``physical association'') \\ \hline
\multirow{3}{*}{Other}
& FAO international trade \cite{dedomenico2015structural} & 364 & $0.13$ & $0.74$ (``Pastry'' and ``Sugar confectionery'') \\
& EU air transportation \cite{cardillo2013} & 37 & $0.03$ & $0.39$ (``Scandinavian Airlines'' and ``Norwegian Air Shuttle'') \\
& ArXiv collaborations \cite{dedomenico2015} & 13 & $0.07$ & $0.73$ (``physics.data-an'' and ``cs.SI'') \\
\hline\hline
\end{tabular}
\label{tab:empiricalCorr}
\end{table*}

To quantify edge correlations at a more granular level, one can first infer block assignments $\vec{g}$ and then calculate correlations $\rho_{rs}$ between all block pairs $(r,s)$. One possible finding from such a calculation may be that correlations between Scandinavian Airlines and Norwegian Air Shuttle routes are significantly larger in certain geographical regions than in others. 


\subsection{Edge Prediction in Shopping Networks}\label{subsec:shopping}

The data-science company dunnhumby gave us access to ``pseudonymized'' transaction data from stores of a major grocery retailer in the United Kingdom. The data were pseudonymized by replacing personally identifiable
information with numerical IDs, rendering it impossible to identify individual shoppers. For our analysis, we aggregate transactions over fixed time windows to construct bipartite networks of customers and products. We refer to these structures as ``shopping networks''. Because some purchases occur in higher volumes than others, it is useful to incorporate edge weights. Given a customer $i$ and a product $j$, the \textit{item-penetration} weight is equal to the fraction of all of the items purchased by customer $i$ that are product $j$. The \textit{basket-penetration} weight is equal to the fraction of all baskets (i.e., distinct shopping trips) of customer $i$ that include product $j$. See the doctoral dissertation \cite{pamfil2018thesis} for more details about these weighting schemes.

\begin{table}[htb!]
\caption{Predictive performance of different models on the shopping data set, as measured by the AUC.}
\label{table:AUCshopping}
\centering
\begin{tabular}{l|c|c}
\hline\hline
\multirow{2}{*}{Model} & AUC & AUC  \\ 
      & ({\sc ShoppingMod}) & ({\sc ShoppingSBM}) \\ \hline
Monolayer SBM   & 0.549 & 0.633 \\
Correlated ER      & 0.724 & 0.743 \\
Correlated SBM     & 0.742 & 0.793 \\ \hline
Monolayer DCSBM & 0.725 & 0.797 \\
Correlated CM      & 0.817 & 0.870 \\
Correlated DCSBM   & 0.818 & 0.875 \\
\hline\hline
\end{tabular}
\end{table}

We now apply the edge-prediction methodology from Sec.~\ref{sec:linkPrediction} to temporal shopping networks, in which edges and edge weights can change from changes in shopping behavior, with a fixed set of customers and a fixed set of products. We construct networks with two layers, which cover the three-month time periods of March--May 2013 and June--August 2013, respectively. {Using the same underlying transaction data, we construct two networks for which we determine the vector $\vec{g}$ of block assignments (the same one for both layers \footnote{{The assumption that communities are identical across layers is a reasonable one for this data set. When performing community detection without this restriction, we find that more than 90\% of nodes stay in the same community across temporal layers. See \cite{pamfil2018thesis} for details.}}) in different ways.} For the first network (which we call {\sc ShoppingMod}), we use basket-penetration weights for the edges and apply multilayer modularity maximization \cite{mucha2010,genlouvain} to the weighted network to determine community assignments $\vec{g}$ \footnote{We use {\sc GenLouvain} \cite{genlouvain} to perform multilayer modularity maximization. This algorithm requires the specification of (at least) two parameters \cite{pamfil2018,mucha2010}, which we set to $\gamma=1.2$ and $\omega=5.0$. {The large value of $\omega$ ensures that community assignments are identical across layers.}}. For the second network (which we call {\sc ShoppingSBM}), we initially calculate item-penetration weights, and we then apply a threshold to remove edges whose weight is below the median weight (i.e., approximately $50\%$ of the edges).
We fit a degree-corrected SBM to the resulting unweighted network using the belief-propagation algorithm from \cite{newman2016annotated}. 
We expect better edge-prediction performance for the second network, because we detect its block structure using an SBM (as opposed to using modularity maximization, which is more restrictive).

\begin{figure*}[htb!]
\centering
\subfloat[{\sc ShoppingMod}, no degree correction]{\hspace{1em}\includegraphics[width=0.4\textwidth]{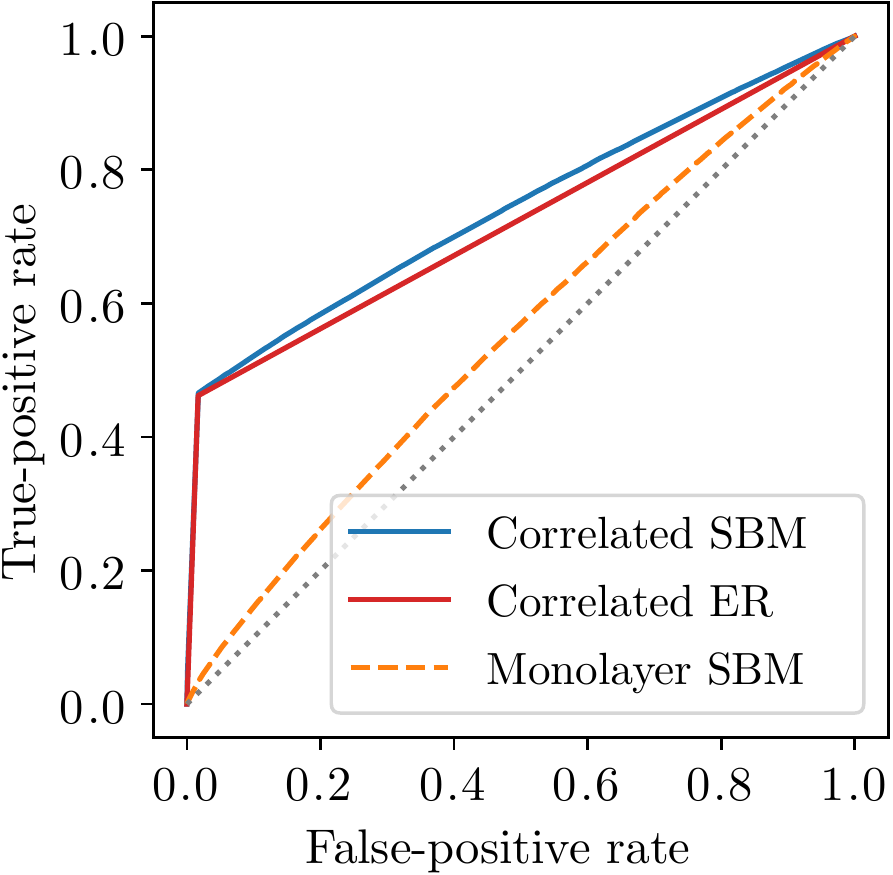}} 
\hspace{4em}
\subfloat[{\sc ShoppingMod}, degree-corrected models]{\includegraphics[width=0.4\textwidth]{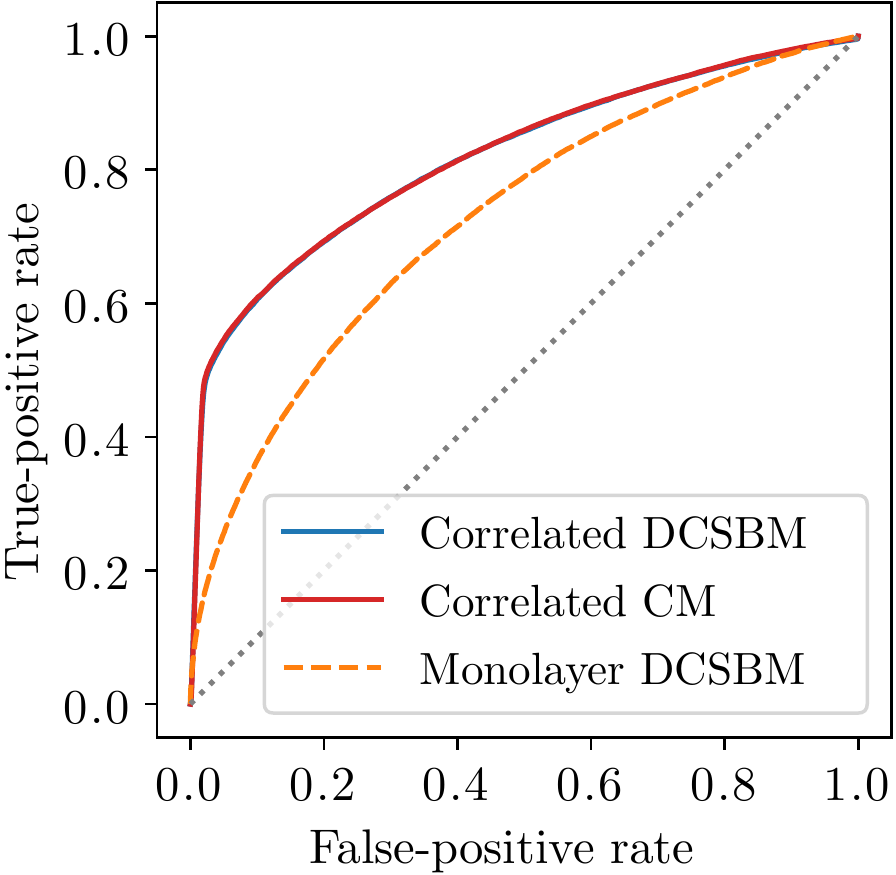}} 
\\
\subfloat[{\sc ShoppingSBM}, no degree correction]{\hspace{1em}\includegraphics[width=0.4\textwidth]{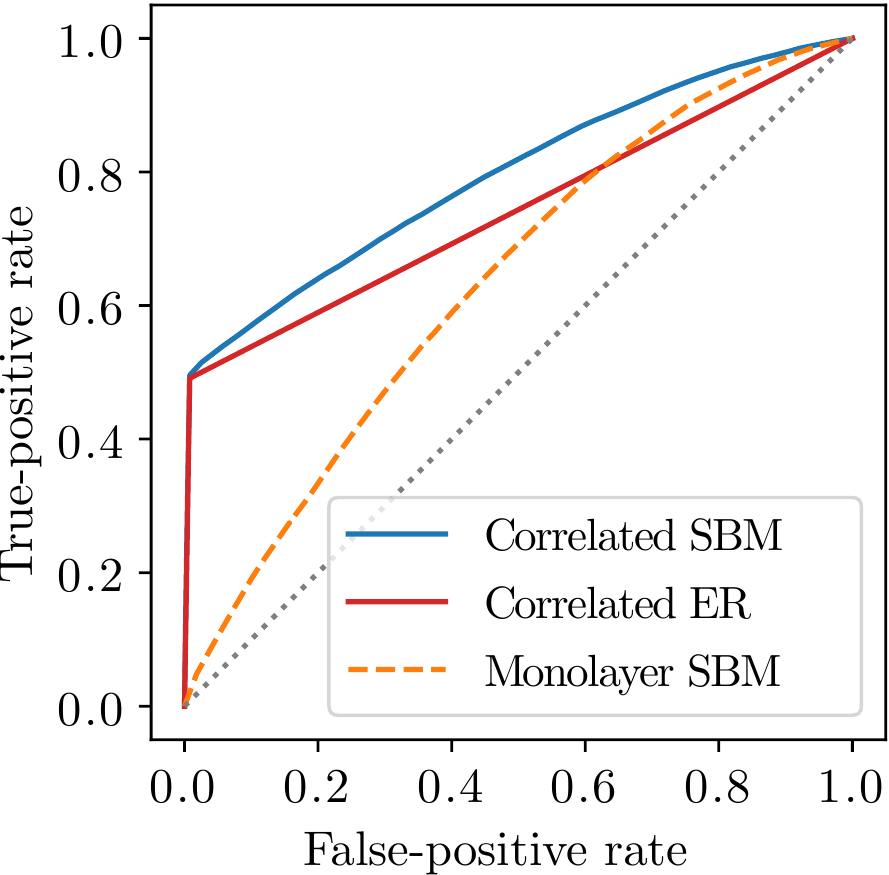}} 
\hspace{4em}
\subfloat[{\sc ShoppingSBM}, degree-corrected models]{\includegraphics[width=0.4\textwidth]{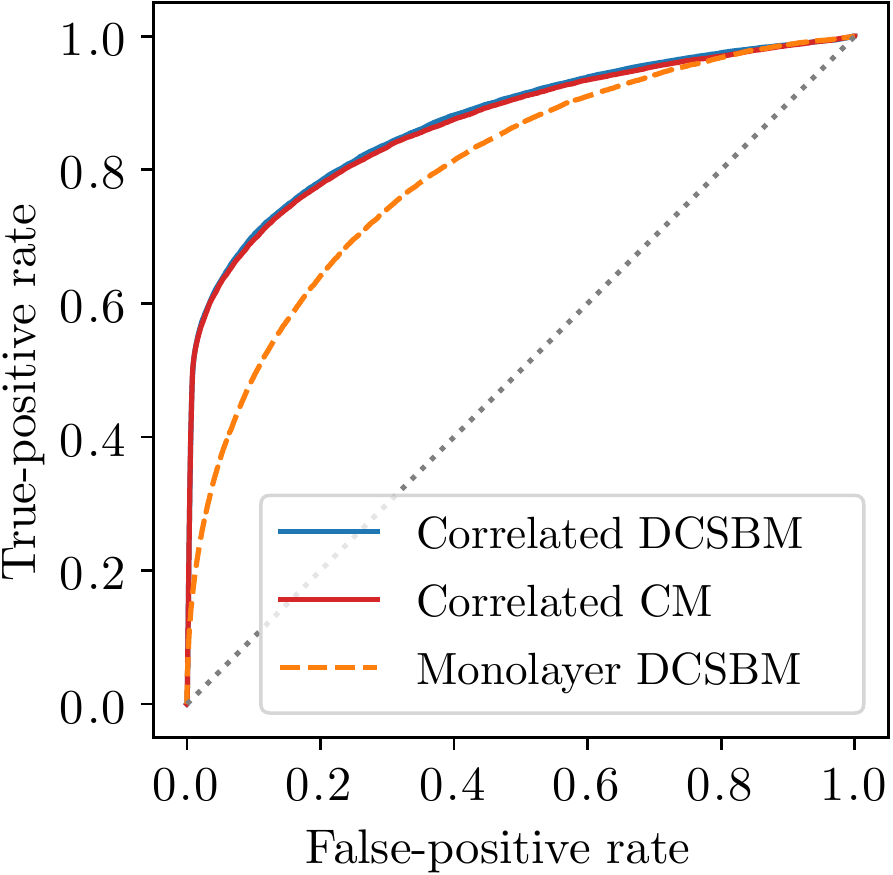}} 
\caption{ROC curves for the edge-prediction task using $5$-fold cross-validation on two temporal networks, {\sc ShoppingMod} and {\sc ShoppingSBM}. We consider both degree-corrected models and models without degree correction. The dotted diagonal line in each plot indicates the expected ROC curve for a random classifier. All other curves lie above this diagonal line, suggesting that they have predictive power. In all cases, correlated multilayer models outperform their monolayer counterparts. For both networks, correlated SBMs outperform correlated ER models, whereas correlated DCSBMs perform similarly to correlated CMs [as illustrated by the almost overlapping curves in panels (b) and (d)].}
\label{fig:rocCurvesShopping}
\end{figure*}

As with our tests on synthetic networks in Sec.~\ref{subsec:edgePredictionSynthetic}, we use $5$-fold cross-validation to assess edge-prediction performance. We summarize the AUC values of our various correlated multilayer models and the monolayer baselines in Table~\ref{table:AUCshopping}, and we show sample ROC curves in Fig.~\ref{fig:rocCurvesShopping}. We make a few observations about these results. First, our approximation $\mathrm{AUC} \approx (1+\rho)/2$ for correlated ER models is very accurate for these two networks, whose correlations are approximately $0.44$ and $0.48$, respectively. Second, our correlated multilayer models outperform the monolayer baselines for both networks. In particular, the very simple correlated ER model --- which assigns one of two probabilities to edges, as indicated in Table \ref{table:edgePredictionSummary} --- performs about as well as the more sophisticated monolayer DCSBM for the {\sc ShoppingMod} network. Third, as expected, AUC values are systematically larger for {\sc ShoppingSBM} than for {\sc ShoppingMod}. Finally, although incorporating mesoscale structure leads to better performance when there is no degree correction, this does not seem to be the case for degree-corrected models, as correlated DCSBMs do not perform significantly better than correlated CMs. This is also apparent in Figs.~\ref{fig:rocCurvesShopping}(b,d), where we observe almost identical ROC curves for the two models. This result suggests that, for some networks, taking into account layer correlations and degree heterogeneity alleviates the need to also consider mesoscale structure when performing edge prediction.
This observation has practical implications, as a correlated CM is much easier than a correlated DCSBM to fit to data and to use for edge prediction. However, for recommendation systems, there are situations in which fitting a correlated DCSBM is beneficial, even if its edge-prediction performance is similar to that of a correlated CM. For instance, one may wish to identify relevant customers for a chosen product, irrespective of how much they buy (i.e., their degree). SBMs are able to distinguish between customers with equal degrees and identify those with the greatest predisposition to buy a particular product, whereas CMs are not.

{We also note a result from \cite{davis2006}
that a curve lies completely above (i.e., ``dominates") another ROC curve if and only if the same relationship holds for the associated PR curves. 
This result implies for almost all of the curves in Fig.~\ref{fig:rocCurvesShopping} that the rankings of our models based on AUCs are almost identical to those that we would obtain if we instead base them on the areas under PR curves. The only curves whose ranking when we use PR curves is unclear from this result are the correlated ER and monolayer SBM curves in Fig.~\ref{fig:rocCurvesShopping}(c).}

{It is informative to consider the source of false positives in Fig.~\ref{fig:rocCurvesShopping}. Because the correlation between layers is positive, we are likely to predict an edge where none exists when all of the following conditions hold: (1) nodes $(i, j)$ are adjacent in one layer but not the other; (2) the node pair $(i, j)$ belongs to an edge bundle $(r, s)$ with a large layer correlation $\rho_{rs}$; and (3) nodes $i$ and $j$ have large degrees. Note that condition (c) applies only for correlated models with degree correction.} 



\section{Conclusions and Discussion}\label{sec:conclusions}

We introduced models of multilayer networks in which edges that connect the same nodes in different layers are not independent. 
In comparison to models {without edge correlations}, our
models offer an improved representation of many empirical networks, as interlayer correlations are a common phenomenon: flights between major airports are serviced by multiple airlines, individuals interact repeatedly with the same people, consumers often buy the same products over time, and so on. Among other potential applications, one can use our models to improve edge prediction, to study the graph-matching problem on more realistic benchmark networks, and to calculate layer correlations as insightful summary statistics for networks.

To model layer correlations, we used bivariate Bernoulli random variables to generate edges simultaneously in two network layers. (See \cite{pamfil2018thesis} for derivations using Poisson random variables.) Correlated Bernoulli stochastic block models were proposed previously \cite{lyzinski2015}, although only as forward models for generating networks, rather than for performing inference given empirical data. Another key contribution of our work is a degree-corrected variant of such a model. The maximum-likelihood equations are significantly more difficult to solve in this case, but we were able to make useful simplifications with suitable approximations. Notably, these simplified equations closely approximate those for models without degree correction for networks with almost homogeneous degree distributions.

The models in the present paper that incorporate some mesoscale structure $\vec{g}$ assume that such structure is given. This setup has the benefit that one can use any desired algorithms to produce a network partition, including ones that operate on weighted or annotated networks or that use nonstandard null models in a modularity objective function. This makes our approach for analyzing correlations suitable for a wide variety of applications.

Fitting a correlated SBM to network data yields a correlation value $\rho_{rs}$ for each edge bundle $(r,s)$. We have defined an effective correlation that combines all of these values into a single measure of similarity between two layers. Notably, the value of the effective correlation is independent of a network's mesoscale structure, making it extremely easy to compute (see Eqn.~\eqref{eqn:rhoERBernoulliMLE}). We illustrated this method of assessing layer similarity for multiplex networks from social, biological, and other domains.

Another application of our work is to edge prediction in multilayer networks. Our numerical experiments revealed that simple correlated models (e.g., a correlated configuration model or a correlated SBM without degree correction) can outperform monolayer DCSBMs {in terms of AUC} even for moderate correlation values. We also observed such improved performance for consumer--product networks, which have significant layer correlations ($\rho \approx 0.45$). We expect that a correlated multilayer DCSBM will typically outperform a monolayer DCSBM for most empirical networks, even when there are lower levels of correlation. 

There are many interesting ways to build further on our work. For example, it would be useful to be able to model all layers simultaneously, rather than in a pairwise fashion, especially for multiplex networks (in which layers do not have a natural ordering). One challenge is that a multivariate Bernoulli distribution of dimension $L$ has $2^L-1$ parameters; this grows quickly with the number $L$ of layers. For a temporal setting, we have proposed generating correlated networks in a sequential way by conditioning each layer on the previous one. In some cases, it will be useful to relax this ``memoryless" assumption and condition a layer on all previous layers, rather than only on the most recent layer. For example, the purchases of shoppers in December one year are strongly related not only to their purchases in November, but also to what they bought in December during the previous year. 

In the present paper, we have not considered the case of nonidentical mesoscale structure across layers. (See \cite{pamfil2018thesis} for a possible approach.) Additionally, although one can account for edge weights when fitting a block structure $\vec{g}$ to use with our models, the rest of our derivations apply only to unweighted networks. Modeling correlated weighted networks entails prescribing both edge-existence and edge-weight correlations.
Yet another idea is to derive correlated models for networks with overlapping communities. Previous research \cite{debacco2017} 
suggests that this can substantially improve edge-prediction performance. 
{For correlated DCSBMs, for example, it would be useful to learn the dependence of the joint probability $\Pp(A_{ij}^1=1,A_{ij}^2=1)$ on node degrees, rather than assume the parametric form in Eqn.~\eqref{eqn:corrDCSBMq}. This is a challenging problem for which maximum-likelihood estimation is unlikely to be a suitable tool, so it falls outside the scope of the present paper.}

Our work is also a starting point for designing algorithms to detect correlated communities in networks by inferring $\vec{g}$ alongside other model parameters. 
A practical outcome of such an algorithm would be a set of communities that persist across layers if and only if the edges in those communities are sufficiently highly correlated with each other. Such an approach would offer a new interpretation of what it means for a community to span multiple layers \cite{bazzi2016,pamfil2018}.

{Lastly, although we performed edge prediction in an unsupervised manner, it is also possible to use estimated correlations as features in supervised models \cite{jalili2017,pujari2015,hristova2016,mandal2018} to improve their performance.}

In summary, our work highlights the importance of relaxing edge-independence assumptions in statistical models of network data. Doing so provides richer insights into the structure of empirical networks, improves edge-prediction performance, and yields more realistic models on which to test community-detection, graph-matching, and other types of algorithms. 


\begin{acknowledgments}

ARP was funded by the EPSRC Centre for Doctoral Training in Industrially Focused Mathematical Modelling (EP/L015803/1) in partnership with dunnhumby. We are grateful to dunnhumby for providing access to grocery-shopping data and to Shreena Patel and Rosie Prior for many helpful discussions. We also thank Junbiao Lu for helpful comments.

\end{acknowledgments}



\appendix

\section{Variance of ML Estimates}\label{app:variance}

{
In Sec.~\ref{sec:corrModels}, we derived ML parameter estimates for various types of correlated network models. In this appendix, we show how to obtain the variance of these estimates and we illustrate how these variances scale with network size (i.e., the number of nodes). 
We focus on correlated ER models and the corresponding log-likelihood \eqref{eqn:corrERloglikelihood}. The same approach also works for the other types of models that we examined. 
}

{
Let $\vec{\beta}=[p_1~p_2~q]^\top$ denote the vector of parameters for a correlated ER model. Under mild conditions \cite{sweeting1980}, the ML estimate $\widehat{\vec{\beta}}$ converges in distribution to a multivariate normal distribution as the number $N$ of nodes tends to infinity. {For large but finite $N$, the quantity $\widehat{\vec{\beta}}$ is distributed approximately according to $\mathcal{N}(\vec{\beta}^*,\mathcal{I}^{-1}(\vec{\beta}^*))$,
where $\vec{\beta}^*$ is the true-parameter vector and $\mathcal{I}^{-1}(\vec{\beta}^*)$ is the inverse of the Fisher information matrix evaluated at the parameter values in $\vec{\beta}^*$.} The diagonal entries of $\mathcal{I}^{-1}(\vec{\beta}^*)$ provide variance estimates for $p_1$, $p_2$, and $q$. 
}

{
To illustrate how these variance estimates scale with the number $N$ of nodes, we simulate correlated ER networks with $p_1=0.1$, $p_2=0.085$, and $q=0.05$ (which correspond to a correlation of $\rho=0.5$) for different network sizes. In Fig.~\ref{fig:varianceFisher}, we plot the 95\% confidence intervals around the ML estimates for each of the three parameters in our model. We find empirically that the variance scales with $1/N^2$, so the standard deviation (and thus the width of the $95\%$ confidence intervals) scales with $1/N$. 
}

\begin{figure}[ht!]
\centering
\includegraphics[width=0.9\columnwidth]{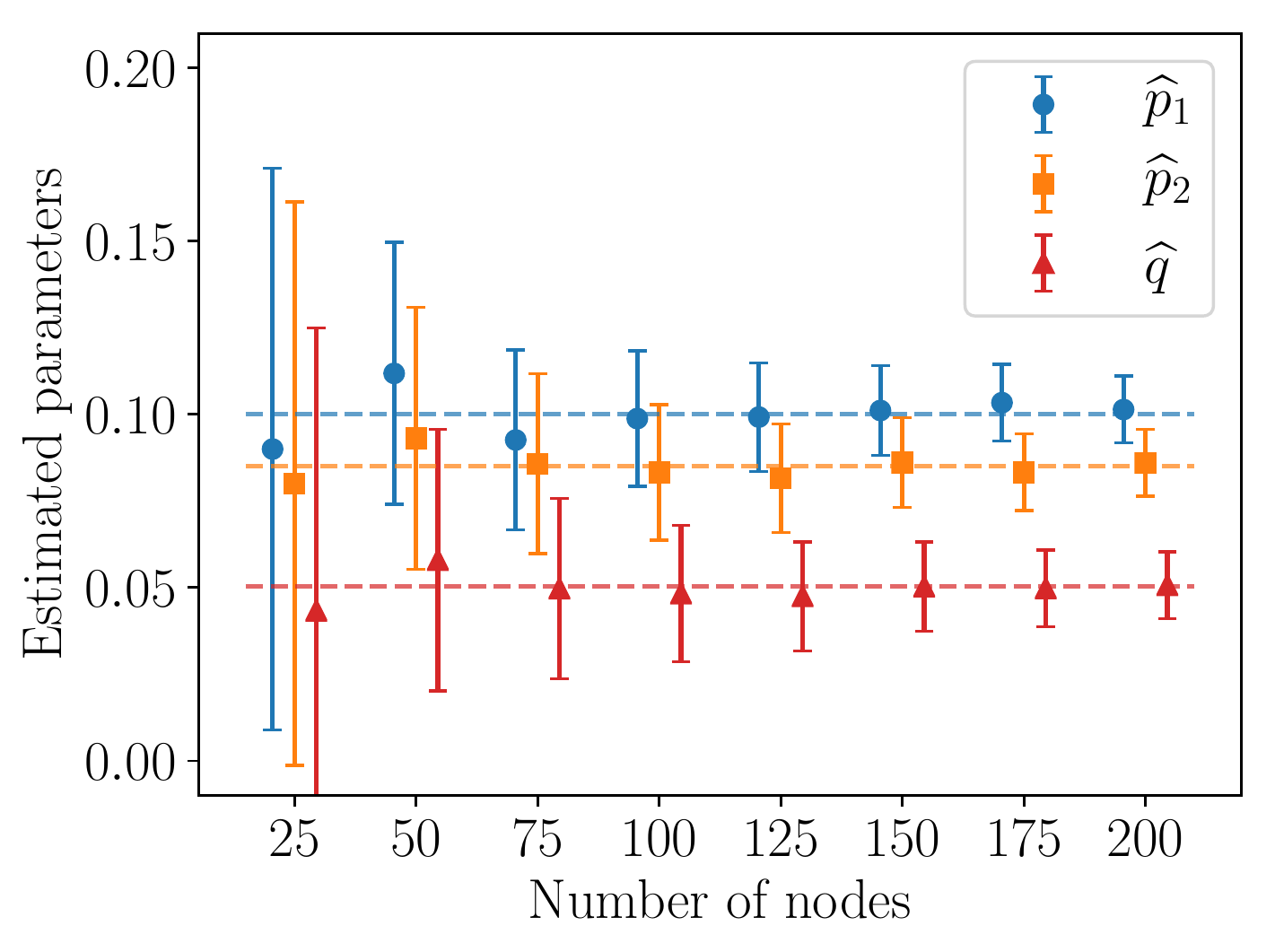}
\caption{ML estimates with $95\%$ confidence intervals that we obtain from the inverse of the Fisher information matrix for different values of the number $N$ of nodes. The true parameter values (dashed horizontal lines) are
$p_1=0.1$, $p_2=0.085$, and $q=0.05$. The width of the confidence intervals scales approximately as $1/N$.
}
\label{fig:varianceFisher}
\end{figure}


\section{Maximizing the Full Log-Likelihood Versus our Approximate Log-Likelihood for Correlated DCSBMs}\label{app:approximation}

{
In Sec.~\ref{subsec:DCSBM}, we derived an approximation \eqref{eqn:BernoulliDC1-3} of the log-likelihood for correlated DCSBMs that enables a more efficient estimation of the parameters than maximizing the exact log-likelihood \eqref{eqn:DCSBMlogL}. We now use simulated data to illustrate how the results that one obtains using this approximation compare with maximizing the original likelihood.
}

{
We simulate networks using the {\sc CorrDCSBM} benchmark (see Sec. \ref{subsec:edgePredictionSynthetic}) with $N=1000$ nodes, $K=5$ communities, a mixing parameter of $\mu=0.3$, and a correlation of $\rho=0.5$. 
For $N=2000$ nodes, which we used in other experiments in the paper, we find that maximizing the full log-likelihood is prohibitively slow. Even for $N=1000$, obtaining estimates from the full log-likelihood takes about 30 minutes in total, whereas the calculation runs in about $5$ seconds on a typical laptop when we use the approximate log-likelihood. 
}

{
The quality of our approximation depends strongly on the shape of the degree distribution: accuracy degrades for distributions with larger variances and heavier tails. We illustrate this behavior in Fig.~\ref{fig:fullVsApproxDCSBM}, where we plot the full-likelihood estimates [see Eqn. \eqref{eqn:DCSBMlogL}] versus the approximate-likelihood estimates [see Eqn.~\eqref{eqn:BernoulliDC1-3}] for two networks. For both networks, we sample degrees from truncated power-law distributions using the code in \cite{benchmark}. In Fig.~\ref{fig:fullVsApproxDCSBM}(a), we choose a relatively narrow degree distribution (with $k_\mathrm{exp}=0$, a minimum degree of $18$, and a maximum degrees of $22$). In this case, the parameter values that we estimate using the approximate log-likelihood closely match those from the full log-likelihood. For Fig.~\ref{fig:fullVsApproxDCSBM}(b), we choose a relatively wide degree distribution (with $k_\mathrm{exp}=-2$, a minimum degree of $10$, and a maximum degree of $50$). In this case, the approximate log-likelihood tends to overestimate parameter values, especially for larger values of these parameters. See the top-right corner of Fig.~\ref{fig:fullVsApproxDCSBM}(b).
}

{
As this example illustrates, there is a trade-off between accuracy and speed. It may be possible to derive better approximations, such as by considering a second-order expansion instead of a first-order expansion with respect to the quantities $\varepsilon_{ij}^1$ and $\varepsilon_{ij}^2$ or by adding corrections for large-degree nodes to the approximate log-likelihood. 
} 

\begin{figure}[ht!]
\centering
\subfloat[$k_\mathrm{exp}=0$, $k_\mathrm{min}=18$, $k_\mathrm{max}=20$]{\includegraphics[height=0.22\textwidth]{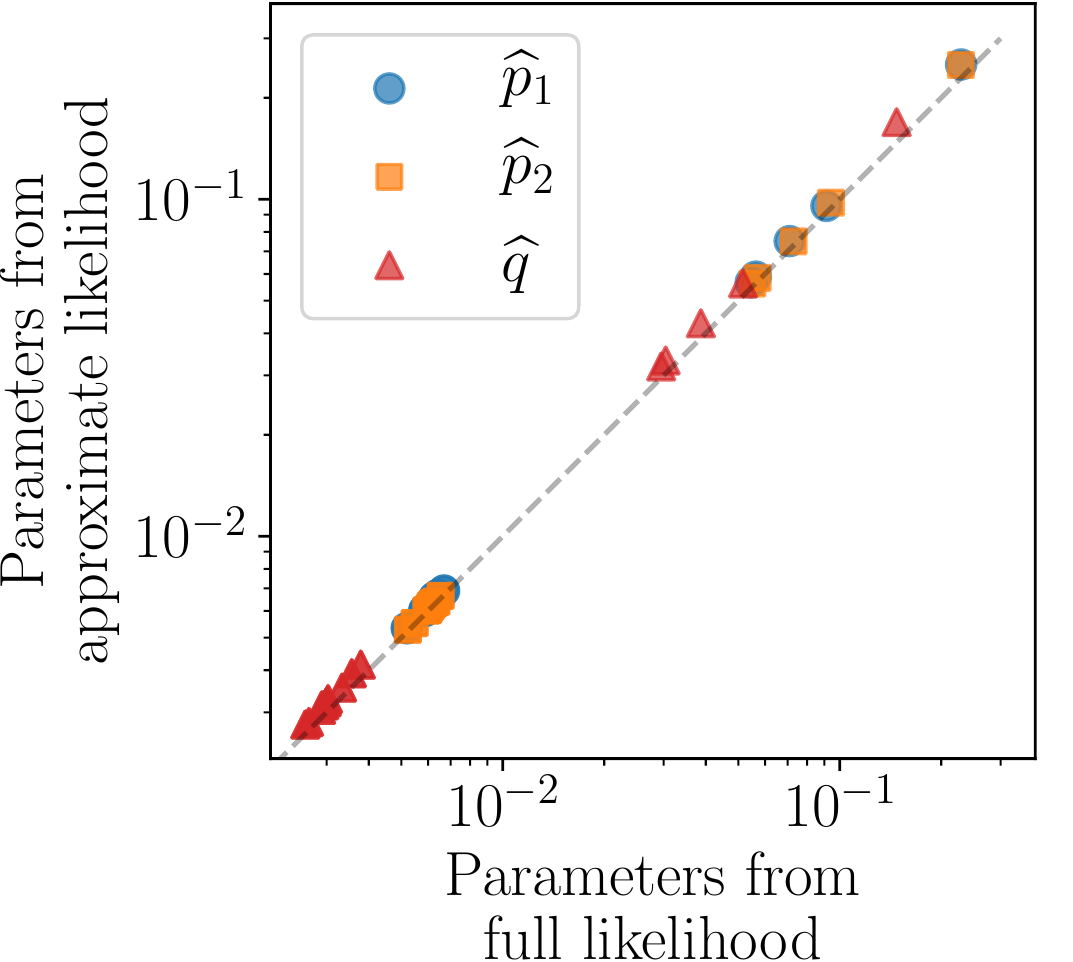}} 
\hspace{0.2em}
\subfloat[$k_\mathrm{exp}=-2$, $k_\mathrm{min}=10$, $k_\mathrm{max}=50$]{\includegraphics[height=0.22\textwidth]{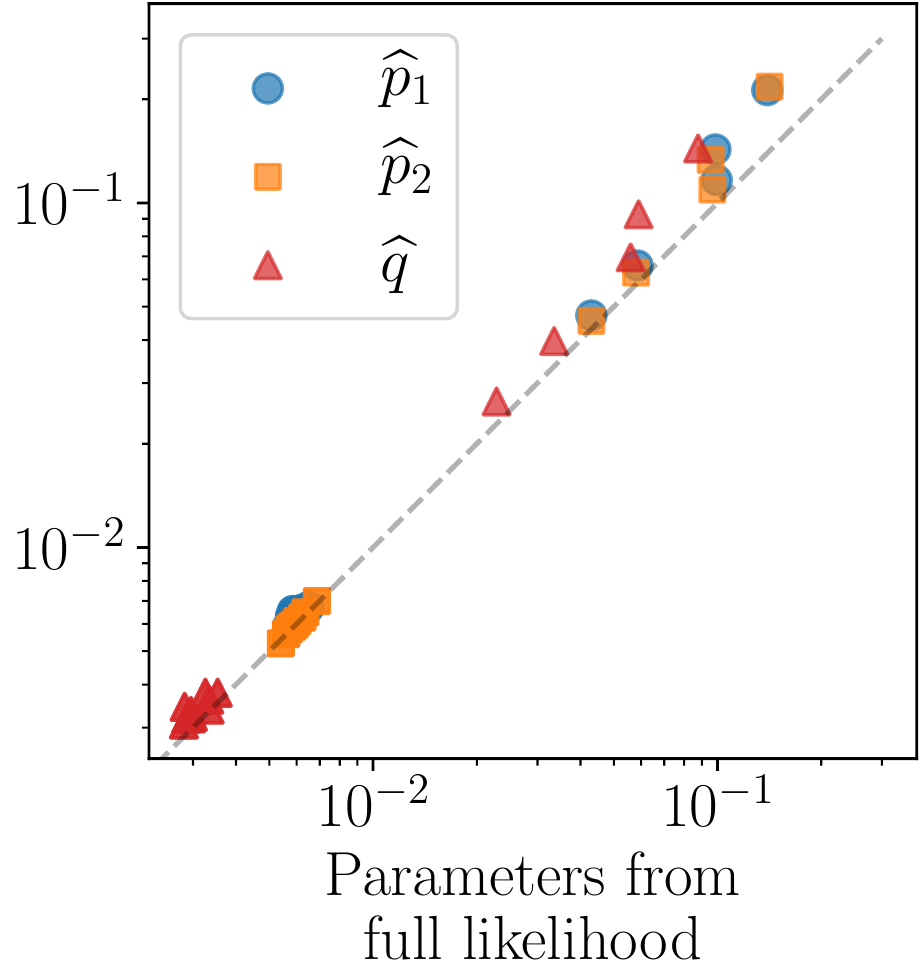}\hspace{0.8em}}
\\
\vspace{0.3em}
\caption{ML estimates from the approximate log-likelihood \eqref{eqn:BernoulliDC1-3} versus ML estimates from the full log-likelihood \eqref{eqn:DCSBMlogL}. (a) Example of a network with a relatively narrow degree distribution. For this example, the approximation works well. (b) Example of a network with a relatively wide degree distribution. For this example, we observe some discrepancies between the two sets of estimates.} 
\label{fig:fullVsApproxDCSBM}
\end{figure}


\section{Edge-Prediction AUC as a Function of the Pearson Correlation $\rho$}\label{app:AUCproof}

We establish the following result.
\begin{proposition}
The AUC for a correlated ER model is an affine function of the Pearson correlation $\rho$. In particular, when $p_1 \approx p_2$ or $p_1 \approx 1-p_2$, we have $\mathrm{AUC}_{\mathrm{ER}} \approx (1+|\rho|)/2$.
\end{proposition}
\begin{proof}
Suppose that $\rho > 0$. (The case $\rho \leq 0$ is similar.) With a correlated ER model, all unobserved interactions $(i,j)$ in the second layer have one of two probabilities: $q/p_1$ if $A_{ij}^1=1$ and $(p_2-q)/(1-p_1)$ if $A_{ij}^1=0$. Because $\rho > 0$, we have $q > p_1p_2$, which implies that $q/p_1 > (p_2-q)/(1-p_1)$. Selecting a threshold between these two probabilities amounts to predicting that everything that is an edge in the first layer is also an edge in the second layer and that everything that is not an edge in the first layer is also not an edge in the second layer. Let $a$ and $b$ denote the TPR and FPR, respectively, at such an intermediate threshold. In this case, $a$ and $b$ are the coordinates of the point at which the slope of the ROC curve changes. See the illustration in Fig.~\ref{fig:AUCDiagram}.

\begin{figure}[ht!]
\includegraphics[width=0.3\textwidth]{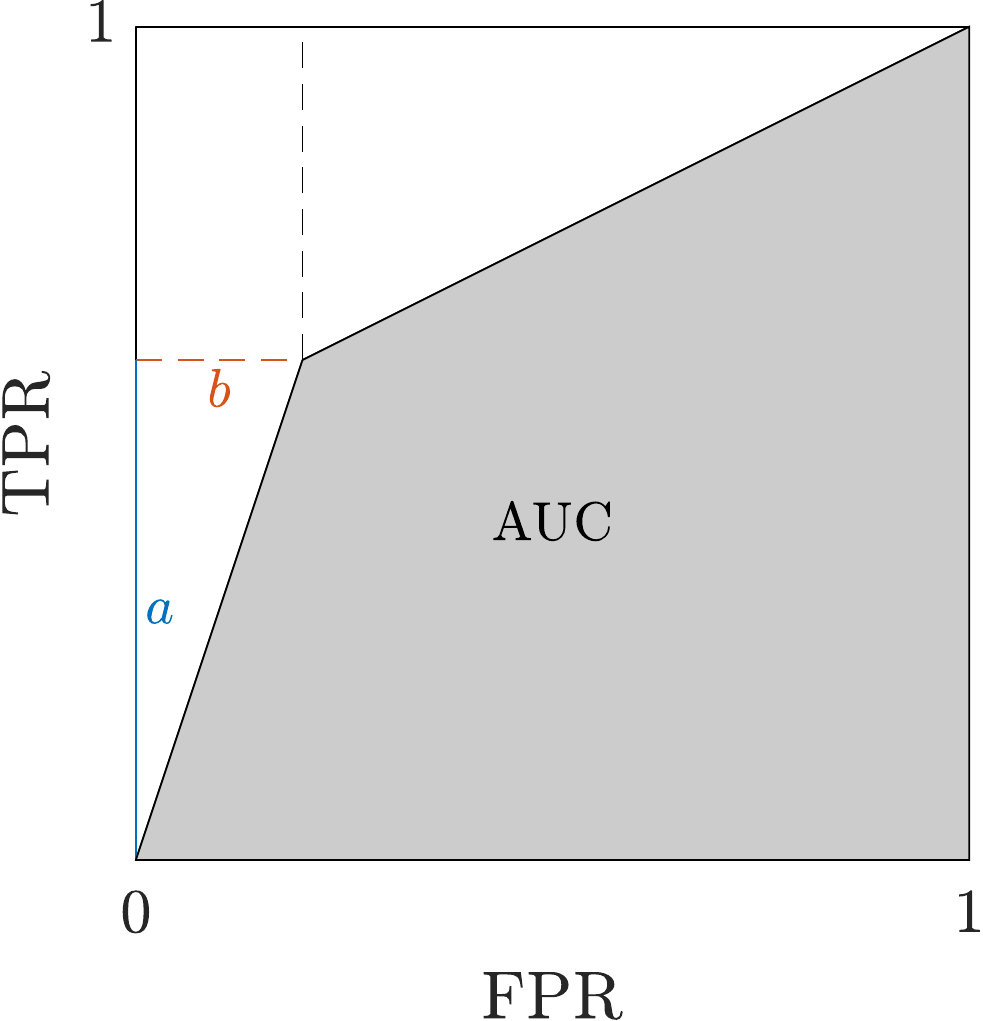}
\caption{Diagram of the ROC curve for an ER model, which assigns one of two possible edge probabilities to each pair of nodes.}
\label{fig:AUCDiagram}
\end{figure}

By straightforward geometry, 
\begin{equation*}
	\mathrm{AUC}_\mathrm{ER}=1-\frac{ab}{2}-\frac{(1-a)(1-b)}{2}-(1-a)b=\frac{1}{2}+\frac{a-b}{2}
\end{equation*}
is the area under this ROC curve. The next step is to estimate $a$ and $b$. {With a correlated ER model}, the number of true positives is proportional \footnote{It is not exactly equal, because with $K$-fold cross validation, we are performing a prediction on only $1/K$ of the data at a time.} to $e_{11}$; the model's prediction is correct every time that an edge that is present in the first layer is also present in the second layer. To find the TPR, one needs to divide this quantity by the number of edges (there are $e_{11}+e_{01}$ of them) in the second layer. 
We obtain
\begin{equation}\label{eqn:ERapproxa}
	a \approx \frac{e_{11}}{e_{11}+e_{01}} \approx \frac{q}{p_2}=p_1+\rho\sqrt{\frac{1-p_2}{p_2}p_1(1-p_1)}\,.
\end{equation}
Similarly, every time an edge that is present in the first layer is not present in the second layer counts as an incorrect prediction of the model. Therefore, the number of false positives is proportional to $e_{10}$. Dividing this by the number of non-edges in the second layer yields
\begin{equation}\label{eqn:ERapproxb}
	b \approx \frac{e_{10}}{e_{10}+e_{00}} \approx \frac{p_1-q}{1-p_2}=p_1-\rho\sqrt{\frac{p_2}{1-p_2}p_1(1-p_1)}\,.
\end{equation}
From \eqref{eqn:ERapproxa} and \eqref{eqn:ERapproxb}, it follows that
\begin{equation*}
	\mathrm{AUC}_\mathrm{ER} \approx \frac{1}{2}+\frac{\rho}{2}\sqrt{p_1(1-p_1)}\left(\sqrt{\frac{1-p_2}{p_2}}+\sqrt{\frac{p_2}{1-p_2}}\right)\,,
\end{equation*}
which is an affine function of $\rho$.
When $p_1 \approx p_2$ or $p_1 \approx 1-p_2$, as is the case in Fig.~\ref{fig:edgePredROC}, we obtain $\mathrm{AUC}_\mathrm{ER} \approx (1+|\rho|)/2$, as desired. Using a similar argument, one can show that the same result holds (with the same assumptions on $p_1$ and $p_2$) when $\rho \leq 0$. 
\end{proof} 
Given that the correlated SBM curves from Fig.~\ref{fig:edgePred} also appear to depend linearly on $\rho$, we believe that it is possible to establish similar results for correlated models that incorporate mesoscale structure.


\section{Data Sets}\label{app:dataSets}

We provide brief descriptions of the multiplex networks that we analyzed in Sec.~\ref{subsec:layerCorr}. For weighted networks, we disregard edge weights when calculating layer correlations. We downloaded these networks, aside from the YouTube and \textit{P. falciparum} data sets, from \url{https://comunelab.fbk.eu/data.php}.


\subsection{CS Aarhus}

This is an undirected and unweighted social network of offline and online relationships between $N=61$ members of the Department of Computer Science at Aarhus University \cite{magnani2013}. There are $T=5$ layers: (1) regularly eating lunch together; (2) friendships on Facebook; (3) co-authorship; (4) leisure activities; and (5) working together.


\subsection{Lazega Law Firm}

This directed, unweighted network encompasses interactions between $N = 71$ partners and associates who work at the same law firm \cite{lazega2001}. The network has $T = 3$ layers that encode co-work, friendship, and advice relationships.


\subsection{YouTube}

This is an undirected, weighted network of interactions between $N=15088$ YouTube users \cite{tang2009}. There are $T=5$ types of interactions: (1) direct contacts (``friendships''); (2) shared contacts; (3) shared subscriptions; (4) shared subscribers; and (5) shared favorites.


\subsection{\textit{C. elegans} Connectome}

This is a directed, unweighted network of synaptic connections between $N=279$ neurons of the nematode \textit{C. elegans} \cite{chen2006}. There are $T=3$ layers, which correspond to electric, chemical monadic (``MonoSyn''), and chemical polyadic (``PolySyn'') junctions. 


\subsection{\textit{P. falciparum} Genes}

This is an undirected, unweighted network of $N=307$ recombinant genes from the parasite \textit{P. falciparum}, which causes malaria \cite{larremore2013}. There are $T=9$ layers that correspond to distinct highly variables regions (HVRs), in which these recombinations occur. 
Two genes are adjacent in a layer if they share a substring whose length is statistically significant. 


\subsection{\textit{Homo sapiens} Proteins}

This is a directed, unweighted network of interactions between $N=18222$ proteins in \textit{Homo sapiens} \cite{dedomenico2015muxviz}. There are $T=7$ layers, which correspond to the following types of interactions: (1) direct interactions; (2) physical associations; (3) suppressive genetic interactions; (4) association; (5) colocalization; (6) additive genetic interactions; and (7) synthetic genetic interactions. The original data is from BioGRID \cite{stark2006}, a public database of protein interactions (for humans as well as other organisms) that is curated from different types of experiments.


\subsection{Food and Agriculture Organization (FAO) Trade}

This is a weighted, directed network of food imports and exports during the year $2010$ between $N=214$ countries \cite{dedomenico2015structural}. There are $T=364$ layers, which correspond to different food products. 


\subsection{European Union Air Transportation}

This is an undirected, unweighted network of flights between $N=450$ airports in Europe \cite{cardillo2013}. There are $T=37$ layers, each of which corresponds to a different airline.


\subsection{ArXiv Collaborations}

This is an undirected, weighted coauthorship network between $N=14489$ network scientists \cite{dedomenico2015}. There are $T=13$ layers, which correspond to different arXiv subject areas: ``physics.soc-ph'', ``physics.data-an'', ``physics.bio-ph'', ``math-ph'', ``math.OC'', ``cond-mat.dis-nn'', ``cond-mat.stat-mech'', ``q-bio.MN'', ``q-bio'', ``q-bio.BM'', ``nlin.AO'', ``cs.SI'', and ``cs.CV''. 
 


%


\end{document}